\documentclass[12pt]{article}

\advance\voffset by -2.0cm \advance\hoffset by -1.25cm
\usepackage{cite}
\usepackage{amsmath}
\usepackage{amssymb}
\usepackage{amsthm}
\usepackage[shortlabels]{enumitem}
\usepackage{authblk}

\usepackage{lscape}
\usepackage{booktabs}
\usepackage{tabularx}
\usepackage{ltablex}
\usepackage{caption}

\usepackage{multirow}
\usepackage[table]{xcolor}

\usepackage[colorlinks,linkcolor=blue]{hyperref}

\newcommand{\CC}{\mathbb{C}} 
\newcommand{\RR}{\mathbb{R}} 
\newcommand{\QQ}{\mathbb{Q}} 
 
\newcommand{\HH}{\mathbb{H}}
\newcommand{\ZZ}{\mathbb{Z}}

\newcommand{\A}{\mathcal{A}}

\newcommand{\C}{\mathcal{T}}
\newcommand{\E}{\mathcal{E}}
\newcommand{\F}{\mathcal{F}}

\newcommand{\M}{\mathcal{M}}
\newcommand{\N}{\mathcal{N}}

\newcommand{\W}{\mathcal{W}}
\newcommand{\Leech}{\Lambda}
\newcommand{\theory}{\mathcal{T}}
\newcommand{\modgrp}{{\mathfrak G}_g}
\newcommand{\modgrpnog}{\mathfrak G}

\newcommand{\Co}{\textsl{Co}}	%Conway group
\newcommand{\cN}{{\mathcal N}}
\newcommand{\eg}{{\cal Z}}

\newtheorem{theorem}{Theorem}
\newtheorem{proposition}[theorem]{Proposition}
\newtheorem{lemma}[theorem]{Lemma}
\newtheorem{corollary}[theorem]{Corollary}

\newtheorem{conjecture}[theorem]{Conjecture}
\newcommand{\SL}{\operatorname{\textsl{SL}}}

\DeclareMathOperator{\Aut}{Aut}

\DeclareMathOperator{\Tr}{Tr}

\newcommand{\gmo}[1]{\Gamma_{\langle-1\rangle}(#1)}

\numberwithin{equation}{section}

\def\be{\begin{equation}}
\def\ee{\end{equation}}

\allowdisplaybreaks[3]
\title{K3 String Theory, Lattices and Moonshine}
\author[1,2]{Miranda C. N. Cheng\thanks{mcheng@uva.nl (On leave from CNRS, France.)}}
\author[3]{Sarah M. Harrison\thanks{sarharr@physics.harvard.edu}}
\author[4,5,6]{Roberto Volpato\thanks{volpato@pd.infn.it}}
\author[5]{Max Zimet\thanks{maxzimet@stanford.edu}}
\date{}
\affil[1]{\small{Korteweg-de Vries Institute for Mathematics, Amsterdam, the Netherlands}}
\affil[2]{\small{Institute of Physics, University of Amsterdam, Amsterdam, the Netherlands}}
\affil[3]{\small{Center for the Fundamental Laws of Nature, 
Harvard University, Cambridge, MA 02138, USA}}
\affil[4]{\small{Dipartimento di Fisica e Astronomia `Galileo Galilei' e INFN sez. di Padova\authorcr
Via Marzolo 8, 35131 Padova (Italy)}}
\affil[5]{\small{Stanford Institute for Theoretical Physics, Department of Physics \authorcr
Stanford University, Stanford, CA 94305, USA}}
\affil[6]{\small{Theory Group, SLAC, Menlo Park, CA 94309, USA\vspace{-15pt}} }
\begin{document}
\maketitle

\abstract{In this paper we address the following two closely related questions. First, we complete the classification of finite symmetry groups of type IIA string theory on $K3 \times \mathbb R^6$, where  Niemeier lattices play an important role.  This extends earlier results by including points in the moduli space with enhanced gauge symmetries in spacetime, or, equivalently, where the world-sheet CFT becomes singular. After classifying the symmetries as abstract groups, we study how they act on the BPS states of the theory. In particular, we classify the conjugacy classes in the T-duality group $O^+(\Gamma^{4,20})$ which represent physically distinct symmetries. Subsequently, we make two conjectures regarding the connection between the corresponding twining genera of $K3$ CFTs and Conway and umbral moonshine, building upon earlier work on the relation between moonshine and the $K3$ elliptic genus.

}
\pagebreak
\tableofcontents

\newpage 

\section{Introduction}
In this paper we study discrete symmetry groups of $K3$ string theory and their action on the BPS spectrum. $K3$ surfaces play an important role in various aspects of mathematics and string theory. For instance, type II string compactifications on $K3\times~T^d \times~\mathbb R^{5-d,1}$ preserve 16 supersymmetries, leading to various exact results regarding the spectrum of BPS states from both the spacetime and world-sheet points of view. In addition, they provide some of the first instances of both holographic duality and a microscopic description of black hole entropy. Geometrically, the Torelli theorem allows for an exact description of the geometric moduli space and makes it possible to analyze the discrete  groups of symplectomorphisms in terms of lattices.  In particular, there is an intriguing connection between $K3$ symmetries and sporadic groups which constitutes the first topic of the current work.

Recall that the sporadic groups are the 26  finite simple groups  that do not belong to any of the infinite families of finite simple groups. Their  exceptional character raises the following questions: Why do they exist? What geometrical and physical objects do they naturally act on? This is one of the reasons why the discovery of (monstrous) moonshine---relating the representation theory of the largest sporadic simple group and a set of canonical modular functions attached to a chiral 2d CFT---is such a fascinating and important chapter in the study of sporadic groups. On the other hand, the relation of other sporadic groups to the ubiquitous $K3$ surface is a surprising result that provides another hint about their true raison d'\^etre.  In this work we will relate  two properties of sporadic groups: moonshine and $K3$ symmetries. 

The connection between $K3$ surfaces and sporadic groups first manifested itself in a celebrated theorem by Mukai \cite{mukai1988finite}, which was further elucidated by Kondo \cite{Kondo}. 
Mukai's theorem established a close relation between the Mathieu group $M_{23}$, one of the 26 sporadic groups, and the symmetries of $K3$ surfaces, in terms of a bijection between (isomorphism classes of) $M_{23}$ subgroups with at least five orbits and (isomorphism classes of) finite groups of $K3$ symplectomorphisms. A generalization of this classical result to ``stringy $K3$ geometry'' was initiated by Gaberdiel, Hohenegger, and Volpato in \cite{K3symm}, using lattice techniques in a method closely following Kondo's proof of the Mukai theorem.  More precisely, the symmetry groups of any non-linear sigma model (NLSM) on $K3$, corresponding to any point in the moduli space \eqref{moduli_space} excepting loci corresponding to singular NLSMs, have been classified in \cite{K3symm}. 
 From the spacetime  (D-branes) point of view, the results of \cite{K3symm} can be viewed as classifying symplectic autoequivalences (symmetries) of derived categories on $K3$ surfaces \cite{huybrechts}. See also \cite{EqK3} for related discussion on symmetries of appropriately defined moduli spaces relevant for curve counting on $K3$.
 The embedding of relevant sublattices of the $K3$ cohomology lattice into the Leech lattice plays an important role in the analysis, and as a result the classification is phrased in terms of subgroups of the automorphism group $\Co_0$ (``Conway zero") of Leech lattice. Recall that there are  24 equivalence classes of 24-dimensional negative-definite even unimodular lattices, called the 24 Niemeier lattices\footnote{Note that this is different from the terminology used in 
\cite{Cheng:2013wca}, where the name ``Niemeier lattice'' is reserved for the twenty-three 24-dimensional negative-definite even unimodular lattices with non-trivial root systems, and hence excludes the Leech lattice.}. All but one of them have root systems of rank 24; these are generated by the lattice vectors of length squared two. The only exception is the Leech lattice, which  has no root vectors. 
 
The first part of the results of the present paper, consisting in a  corollary (Corollary \ref{cor_in_phys_terms}) of  two mathematical theorems (Theorem \ref{t:dirclass} and \ref{t:invclass}), extends this classification to theories corresponding to singular loci in the moduli space of $K3$ NLSMs. 
It is necessary to make use of all  24 Niemeier lattices in order to generalize the analysis to include these singular loci. 
Despite the fact that the type IIA worldsheet theory behaves badly along these loci \cite{Aspinwall:1995zi}, the full type IIA string theory is not only completely well-defined but  also possesses special physical relevance in connection to non-Abelian gauge symmetries. Recall 
 that the spacetime gauge group is enhanced from  $U(1)^{24}$ to some nonabelian group  at these loci, and the ADE type gauge group  is given by the ADE type singularity of the $K3$ surface \cite{Witten:1995ex,Aspinwall:1995zi}. The existence of such loci with enhanced gauge symmetries  in the moduli space, though not immediately manifest from the world-sheet analysis in type IIA, is clear from the point of view of the dual heterotic $T^4$ compactification. In this work we are interested in finite group symmetries which preserve the ${\cal N}=(1,1)$ spacetime supersymmetry from the point of view of  type IIA compactifications.

Apart from these physical considerations, another important motivation to understand the discrete symmetries of general type IIA compactifications on $K3$ surfaces is the following. 
The $K3$ surface--sporadic group connection has  recently entered the spotlight due to the discovery of new moonshine phenomena, initiated by an observation of Eguchi, Ooguri, and Tachikawa (EOT) \cite{EOT}. The $K3$ elliptic genus \eqref{defEG}, a function which counts BPS states of  $K3$ NLSMs and a loop-space index generalizing the Euler characteristic and the $A$-roof genus, is shown to encode an infinite-dimensional graded representation of the largest Mathieu sporadic group $M_{24}$. (Note that the group featured in Mukai's theorem, $M_{23}$, is a subgroup of $M_{24}$ as the name suggests.)
A natural guess is hence that there exists a $K3$ NLSM with $M_{24}$ acting as its symmetry group. However, the classification result of \cite{K3symm} precludes this solution, and one must find an alternative way to explain  Mathieu moonshine. See \S\ref{sec:discussion} for further discussion on this point.

The observation of EOT was  truly surprising  and led to a surge in activity in the study of (new) moonshine phenomena. 
Two of the subsequent developments, regarding umbral and Conway moonshines and their relation to $K3$ NLSMs, motivated the second part of our results which are encapsulated by two conjectures (Conjecture \ref{conj:all_arise_from_moonshine} and \ref{c:manytwin}) and further detailed in appendix \ref{a:classifclasses}.

 The first  development is the discovery of umbral moonshine and its proposed relation to stringy $K3$ geometry. A succinct and arguably the most natural way to describe Mathieu moonshine is in terms of the relation between a certain set of mock modular forms and $M_{24}$. See, for instance, \cite{DMZ} for an introduction on mock modular forms. 
Studying Mathieu moonshine from this point of view \cite{Cheng:2011ay}, it was realized in  \cite{Cheng:2012tq,Cheng:2013wca} that it is but one case of a larger structure, dubbed  umbral moonshine. 
Umbral moonshine consists of a family of 23 moonshine relations corresponding to the 23 Niemeier lattices $N$ with non-trivial root systems: while the automorphism group of a Niemeier lattice dictates the relevant finite group $G_N$ (cf. \eqref{Niemeier_grp}), the root system of the lattice helps determine a unique (vector-valued) mock modular form associated with each conjugacy class of $G_N$. See \S\ref{s:conjectures} for  more detail. One of the umbral moonshine conjectures then states that there exists a natural way to associate a graded infinite-dimensional module with the finite group $G_N$ such that its graded character coincides with the specified mock modular forms. So far, these modules have been shown to exist \cite{Gannon:2012ck,Duncan:2015rga}, although,  with the exception of a special case  \cite{Duncan:2014tya}, their construction is still lacking. While the mock modularity suggests a departure from the usual vertex operator algebra (VOA; or chiral CFT) structure inherent in, e.g., monstrous moonshine, the existence of the generalized umbral moonshine  \cite{Gaberdiel:2012gf,Cheng:2016nto} suggests that certain key features of VOA should nevertheless be present in the modules underlying umbral moonshine. 
 Subsequently, motivated by previous work \cite{Ooguri:1995wj,nikulin2013kahlerian}, the relation between all 23 instances of umbral moonshine and symmetries of $K3$ NLSMs was suggested in  \cite{Cheng:2014zpa} in the form of a proposed relation  \eqref{um_twining} between the umbral moonshine mock modular forms and the $K3$ elliptic genus twined by certain symmetries \eqref{twiningdef}.
  
The second important development, inspired by the close relation between the Conway group $\Co_0$ and  stringy $K3$ symmetries \cite{K3symm}, relates Conway moonshine also to the twined $K3$ elliptic genus \cite{duncan2016derived}. The Conway moonshine module is a chiral superconformal field theory with $c=12$ and symmetry group $\Co_0$, which was first discussed in \cite{frenkel1985moonshine} and further studied in \cite{duncan2007super,Duncan:2014eha}. Using the Conway module, the authors of \cite{duncan2016derived} associate two (possibly coinciding) Jacobi forms to each conjugacy class of $\Co_0$, and conjecture that this set constitutes a complete list of possible $K3$ twining genera. In particular, it was conjectured that one of the two such Jacobi forms arising from Conway moonshine is attached to each symmetry of any non-singular $K3$ NLSM. Note that many, but not all, of the functions arising from umbral moonshine \cite{Cheng:2014zpa} and Conway moonshine \cite{duncan2016derived} coincide. 

As the first part of our results establishes the importance of all 24 Niemeier lattices in the study of symmetries of $K3$ string theory, it is natural to suspect that both umbral and Conway moonshine might play a role in describing the action of these symmetry groups on the (BPS) spectrum of  $K3$ string theory. Note that the CFT is not well-defined at the singular loci of the module space, and hence we restrict our attention to the non-singular NLSMs when we discuss the (twined) elliptic genus. Motivated by the connection between the stringy K3 symmetries and moonshine,  our analysis of world-sheet parity symmetries of NLSMs (see \S\ref{s:parity}), and  results regarding   Landau--Ginzburg orbifolds \cite{cheng2015landau}, in this paper we conjecture (Conjecture \ref{conj:all_arise_from_moonshine}) that the proposed twining genera arising from umbral and Conway moonshine as defined in  \cite{Cheng:2014zpa} and  \cite{duncan2016derived}
 capture all of the possible discrete stringy symmetries of any NLSM in the $K3$ CFT moduli space. Moreover, we conjecture (Conjecture \ref{c:manytwin}) that each of the umbral and Conway moonshine functions satisfying certain basic assumptions (that the symmetry preserves at least a four-plane in the defining 24-dimensional representation) is realized as the physical twining genus of a certain $K3$ NLSM. These conjectures pass a few non-trivial tests. In particular, in this paper we also obtain an almost complete classification of conjugacy classes of the discrete T-duality group $O^+(\Gamma^{4,20})$, as well as a partial classification of the twined $K3$ elliptic genus using methods independent of moonshine. These classification results, summarized in Table \ref{last_table_genera}, are not only of interest on their own but also provide strong evidence for these conjectures which consolidate our understanding of stringy $K3$ symmetries and the relation between $K3$ BPS states and moonshine. 
  
 The rest of the paper is organized as follows. In \S \ref{sec:class}, we classify the 
symmetry groups which arise in type IIA string theory on $K3\times \mathbb R^6$ and preserve the world-sheet $\mathcal N=(4,4)$ superconformal algebra in terms of two theorems. 
This extends the result of \cite{K3symm} to singular points in the moduli space of $K3$ NLSMs. 
In \S\ref{sec:conj} we discuss how these symmetry groups  act on the BPS spectrum of the theory. In particular, we present two conjectures relating the twining genera of NLSMs to the functions which feature in umbral and Conway moonshine. 
In \S \ref{sec:alltwinings} we summarize all the computations of twining genera in physical models that are known so far, including torus orbifolds and Landau-Ginzburg orbifolds, and explain how this data provides evidence for our conjectures.
 Finally, we conclude with a discussion in \S \ref{sec:discussion}. 
A number of appendices include useful information which complements the main text. 
In appendix \ref{app:lats} we summarize some basic facts about lattice theory. The proofs of our main theorems discussed in section \ref{sec:class} can be found in appendix \ref{app:pfs}. In appendix \ref{a:modularity} we present the arguments that we employ in \S\ref{sec:conj} to determine the modular properties of certain twining genera.  In appendix \ref{a:classifclasses} we discuss the method we use to classify distinct $O^+(\Gamma^{4,20})$ conjugacy classes. The result of the classification, as well as the data of the twining genera, are recorded in Table \ref{last_table_genera}.

\section{Symmetries}\label{sec:class}

In this section, we classify subgroups of $O^+(\Gamma^{4,20})$ that pointwise fix a {\it positive four-plane},  a four-dimensional oriented positive-definite subspace of $\Gamma^{4,20} \otimes_{\ZZ} \RR$. 
They have the physical interpretation as groups of supersymmetry-preserving discrete symmetries of type IIA string theory on $K3\times \RR^6$. Alternatively, they can be viewed as the symmetry groups of
 NLSMs on $K3$ surfaces that commute with the 
 $\N=(4,4)$ superconformal algebra and leave invariant the four R-R ground states corresponding to the spectral flow generators.

We will say such $G\subset O^+(\Gamma^{4,20})$ is a subgroup of {\it four-plane preserving type}, and denote the corresponding invariant and co-invariant sublattices by
\be
\Gamma^G:=\{v\in \Gamma^{4,20}\mid g(v)=v\text{ for all } g\in G\}~,~\Gamma_G:=(\Gamma^G)^\perp\cap \Gamma^{4,20}. 
\ee
Note that such a group of four-plane preserving type can in general preserve more than just a four-plane, for instance the trivial group. 

Our result extends \cite{K3symm} by allowing the co-invariant lattice to contain root vectors. Namely, we include those subgroups of four-plane preserving type such that there exists a $v\in \Gamma_G$ with $\langle v,v \rangle =-2$, where  $\langle \cdot, \cdot\rangle $ denotes the bilinear form of the lattice $\Gamma^{4,20}$. 
We say that a positive four-plane is a {\it singular positive four-plane} if it is orthogonal to some root vector. 
Physically, they correspond to type IIA compactifications with  enhanced gauge symmetry, or to singular NLSMs. 
The 23 Niemeier lattices with roots play an important role in the analysis of these singular cases.

\subsection{The Moduli Space}

Let us first review some general properties of NLSMs on K3 (see \cite{Aspinwall:1996mn,Nahm:1999ps}). The moduli space of NLSMs on $K3$ with $\N=(4,4)$ supersymmetry is given by
\be \label{moduli_space}\M=(SO(4)\times O(20))\backslash O^+(4,20)/O^+(\Gamma^{4,20})\ ,
\ee where $(SO(4)\times O(20))\backslash O^+(4,20)$ is the Grassmannian of positive four-planes $\Pi$ within $\RR^{4,20}\cong \Gamma^{4,20}\otimes_\ZZ \RR$, and $\Gamma^{4,20}$ is the even unimodular lattice with signature $(4,20)$. 
This is also the moduli space of type IIA string theory at a fixed finite value of $g_s$. 

The real group $O(4,20) := O(4,20;\RR)$ has four connected components 
\be O(4,20)=O^{++}(4,20)\cup O^{+-}(4,20)\cup O^{-+}(4,20)\cup O^{--}(4,20)\ ,
\ee where the elements of 
\be O^+(4,20)=O^{++}(4,20)\cup O^{+-}(4,20)\ \ee preserve the orientation of positive four-planes.\footnote{Here and in the following, by ``orientation of positive four-planes" we mean each of the two equivalence classes of oriented positive four-planes in $\RR^{4,20}$ modulo $O(4,20)$ transformations connected to the identity.}
 We denote by $O(\Gamma^{4,20})\subset O(4,20)$ the group of automorphisms of the lattice $\Gamma^{4,20}$ and define
 \be O^+(\Gamma^{4,20})=O(\Gamma^{4,20})\cap O^+(4,20)\ .
 \ee

In this work, the lattice $\Gamma^{4,20}$ plays the following roles. Geometrically, $\Gamma^{4,20}$ is the integral cohomology lattice $H^\ast(X,\ZZ)$ with Mukai pairing of a $K3$ surface $X$, and $\Gamma^{4,20}\otimes_\ZZ \RR$ is the real cohomology. Physically, $\Gamma^{4,20}$ is the lattice of D-brane charges, $\Gamma^{4,20}\otimes_\ZZ \RR$ is the space of R-R ground states, and $\Pi\subset \Gamma^{4,20}\otimes_\ZZ \RR$ is the subspace spanned by the four spectral flow generators, i.e. the R-R ground states which furnish a $({\bf 2},{\bf 2})$ representation of the $SU(2)_L\times SU(2)_R$ R-symmetry group. 
From the point of view of the spacetime physics, the choice of a positive four-plane $\Pi$ is given by a choice of the (spacetime) central charge 
 $Z: \tilde H^{1,1}(X,\ZZ) \to \CC$, where $\tilde H^{1,1}(X,\ZZ):= H^{0,0}(X,\ZZ)\oplus H^{1,1}(X,\ZZ)\oplus H^{2,2}(X,\ZZ)$, which determines the mass of supersymmetric D-branes. %

Note that in the existing literature the moduli space is often defined as the quotient of the Grassmannian by the full automorphism group $O(\Gamma^{4,20})$ instead of $O^+(\Gamma^{4,20})$.
 As we explain in more detail in \S \ref{s:parity}, dividing by $O^+(\Gamma^{4,20})$ amounts to distinguishing between NLSMs that are related by world-sheet parity \cite{Nahm:1999ps}. 
Due to the existence of symmetries that act differently on the right- and left-moving states of the NLSM, it is crucial for us to identify $O^+(\Gamma^{4,20})$ instead of $O(\Gamma^{4,20})$ as the relevant group of duality.

\subsection{Symmetry Groups}

 Let us denote by $\C(\Pi)$ the NLSM associated to a given non-singular positive four-plane $\Pi$.  
With some abuse of notation, we will use the same letter for the lattice automorphism $h\in O^+(\Gamma^{4,20})$ and the corresponding duality between the two CFTs $\C(\Pi)$ and $\C(\Pi')$, where $\Pi' := h(\Pi)$. 
 Let  $G$ be the group of symmetries of a non-singular NLSM $\C(\Pi)$ preserving the $\N=(4,4)$ superconformal algebra and  the four spectral flow generators. 
It is shown in \cite{K3symm} that $G$ is given by the largest $O(\Gamma^{4,20})$-subgroup whose induced action on $\Gamma^{4,20}\otimes_\ZZ \RR$ fixes $\Pi$ point-wise, and hence is always a subgroup of $O^+(\Gamma^{4,20})\subset O(\Gamma^{4,20})$.
From the space-time point of view, the group $G$  admits  the alternative interpretation as the spacetime-supersymmetry-preserving discrete symmetry group of a six-dimensional type IIA string theory with half-maximal supersymmetry, away from the gauge symmetry enhancement points in the moduli space. 
More precisely, $G$ is the 
 group of symmetries commuting with all space-time supersymmetries, 
 quotiented by its continuous (gauge) normal subgroup $U(1)^{24}$.

When $\Pi$ is a singular four-plane,
the NLSM $\C(\Pi)$ 
  is not well-defined
  and it is hence meaningless to talk about the symmetry group of the NLSM in this case. 
On the other hand, note that the two alternative definitions of the symmetry group -- $G$ as the point-wise stabilizer of the subspace $\Pi$ and as the discrete symmetry group of type IIA string theory -- can be extended to singular four-planes
without any difficulty. 
  One subtlety, however, is that the two definitions are not equivalent for singular models. 
  Indeed, the pointwise stabilizer group ${\rm Stab}(\Pi)$ of $\Pi$ contains a normal subgroup $W\subseteq {\rm Stab}(\Pi)$ which is the Weyl group corresponding to the set of roots $v\in \Gamma^{4,20}$ orthogonal to $\Pi$. On the other hand, in the type IIA compactification, this Weyl group $W$ is part of the continuous (non-abelian) gauge group, and therefore it is quotiented out in the definition of the 
spacetime discrete symmetry group $G_{\rm IIA}$, i.e.
\be\label{twodefs} G_{\rm IIA}={\rm Stab}(\Pi)/W\ .
\ee
While $G_{\rm IIA}$ is the most interesting group from the point of view of string theory, the group ${\rm Stab}(\Pi)$ admits a more direct mathematical definition. Furthermore, by \eqref{twodefs}, it is straightforward to recover $G_{\rm IIA}$ once ${\rm Stab}(\Pi)$ is known. As a result, we will mainly focus on ${\rm Stab}(\Pi)$ in this section. In terms of the symplectic autoequivalencies of the bounded derived category ${\rm D}^{\rm b}({\rm Coh}(X))$ of coherent sheaves of a $K3$ surface $X$, allowing for the orthogonal complement of $\Pi$ in $\Gamma^{4,20}$ to contain roots amounts to relaxing the stability condition in \cite{MR2376815,huybrechts} to allow for the central charge $Z: \tilde H^{1,1}(X,\ZZ) \to \CC$ to vanish on some $\delta$ with $\delta^2=-2$.

Let us now consider the problem of classifying all $O^+(\Gamma^{4,20})$  subgroups of four-plane fixing type, including those involving singular four-planes. 
Notice that by definition the invariant lattice $\Gamma^G$ has signature $(4,d)$ for some $0\le d\le 20$, and hence the co-invariant lattice $\Gamma_G$ is negative-definite of rank $20-d$.

In \cite{K3symm}, it is shown that if $\Gamma_G$ contains no roots then it can be primitively embedded into the Leech lattice ${\Leech}$ (taken negative definite) 
\be i:\Gamma_G\hookrightarrow {\Leech}
\ee
and that $G$ is isomorphic to a subgroup $\hat G$ of the Conway group $\Co_0\cong O({\Leech})$. More precisely, $\hat G\subset \Co_0$ acts faithfully on $\Lambda_{\hat G}:=i(\Gamma_G)\subset {\Leech}$ and fixes pointwise the orthogonal complement ${\Lambda}^{\hat G} = ({\Lambda}_{\hat G})^\perp\cap {\Leech}$.  

In order to generalize the classification of the symmetry groups $G$ to singular four-planes, we have to consider the case where $\Gamma_G$ contains a root.
It is clear that in this case, lattices with non-trivial root systems---i.e. Niemeier lattices other than the Leech lattice---are necessary for the embedding. 
In fact, in this case the co-invariant lattices can be always embedded into one of the Niemeier lattices, as we show with the following theorem. 

\begin{theorem}\label{t:dirclass}
Let $G$ be  a subgroup of $O^+(\Gamma^{4,20})$ fixing pointwise a sublattice $\Gamma^G$ of signature $(4,d)$, $d\ge 0$. Then there exists a primitive embedding $i$ of the orthogonal complement $\Gamma_G$ into some negative-definite rank 24 even unimodular lattice (Niemeier lattice) $N$ 
\be \label{e:embed} i:\Gamma_G\hookrightarrow N\ .
\ee Furthermore, the action of $\hat G:=iGi^{-1}$ on $i(\Gamma_G)$ extends uniquely to a group of automorphisms of $N$ that fixes pointwise the orthogonal complement of $i(\Gamma_G)$. If $\Gamma_G$ has no roots, then $N$ can be chosen to be the Leech lattice.
\end{theorem}

\begin{proof}
See appendix \ref{a:proofdirclass}.
\end{proof}

Note that the embedding is generically far from unique, and often $\Gamma_G$ can be embedded in more than one Niemeier lattice $N$. 
At the same time, we believe that all Niemeier lattices are necessary in order to embed all $\Gamma_G$ as in (\ref{e:embed}). In particular, in a geometric context it was conjectured in \cite{nikulin2013kahlerian} that for each of the 24 Niemeier lattices $N$ there exists a (non-algebraic) $K3$ surface $X$ whose Picard lattice $P(X)$ can be primitively embedded only in $N$. This conjecture has been proven for all but two Niemeier lattices: those with root systems $A_{24}$ and $2A_{12}$. It is possible to find an appropriate choice of the B-field such that the orthogonal complement lattice $\Gamma_G$ contains the Picard lattice.  Therefore, we expect all Niemeier lattices (and not just the Leech lattice) play a role in the study of physical symmetries of type IIA string theory on $K3$. 

By theorem \ref{t:dirclass}, every group of symmetries $G$ is isomorphic to a subgroup $\hat G\subset O(N)$ of the group of automorphisms of some Niemeier lattice $N$, fixing a sublattice of $N$ of rank at least $4$. In fact, the converse is also true by the following theorem.

\begin{theorem}\label{t:invclass}
Let $N$ be a (negative definite) Niemeier lattice and $\hat G$ be  a subgroup of $O(N)$ fixing pointwise a sublattice $N^{\hat G}$ of rank $4+d$, $d\ge 0$. Then, there exists a primitive embedding
\be f:N_{\hat G}\hookrightarrow \Gamma^{4,20}
\ee
 of the co-invariant sublattice $N_{\hat G}:=(N^{\hat G})^\perp\cap N$ into the even unimodular lattice $\Gamma^{4,20}$.
 Furthermore, the action of $ G:=f{\hat G} f^{-1}$ on $f(N_{\hat G})$ extends uniquely to a group of automorphisms of $\Gamma^{4,20}$ that fixes pointwise the orthogonal complement of $f(N_{\hat G})$. Therefore, there exists a positive four-plane $\Pi$ such that ${\rm Stab}(\Pi)$ contains $\hat G$ as a subgroup.
  When $N$ is the Leech lattice, $\Pi$ can be chosen so that its orthogonal complement contains no roots.
\end{theorem}
\begin{proof} See appendix \ref{a:proofinvclass}.
\end{proof}

As we will discuss in the next subsection, for  many $G$ arising in the way described above, there exist continuous families of $\Pi$ such that the above statement is true, while for those groups with  invariant sublattice of rank exactly four, the family consists of isolated points.

It is now useful to make a comparison to the groups in umbral and Conway moonshine (cf. \S\ref{s:conjectures}).
 When $N$ is a Niemeier lattice with roots, the automorphism group $O(N)$ contains  as a normal subgroup the Weyl group $W_N$, generated by reflections with respect to the hyperplanes orthogonal to the roots. The quotients 
 \begin{equation} \label{Niemeier_grp}G_N:=O(N)/W_N\end{equation} are the groups whose representation theory dictates the mock modular forms featuring in umbral moonshine \cite{Cheng:2013wca}. 
 To uniformize the notation, when $N=\Leech$ is the Leech lattice, we define $W_N$ to be the trivial group and subsequently $G_N=O(N)=\Co_0$. We will refer to these $G_N$ as the {\it Niemeier groups}.
Next we discuss the properties of $G$ in relation to the Niemeier groups.

\begin{proposition}\label{t:UmbralSymm}
For a given sublattice $\Gamma^G\subset \Gamma^{4,20}$ of signature $(4,d)$ with $d\ge 0$, let 
$G:=\{g\in O(\Gamma^{4,20})\lvert gv=v~~ \forall v\in \Gamma^G\}.$
Suppose that the orthogonal complement $\Gamma_G$ can be primitively embedded in the Niemeier lattice $N$, so that $G$ is isomorphic to a subgroup $\hat G\subset O(N)$. Then:
\begin{enumerate}
\item $\hat G$ has non-trivial intersection with the Weyl group $W_N$ if and only if $\Gamma_G$ contains some root.
\item if $\Gamma_G$ has no roots, then $G$ is isomorphic to a subgroup of $G_N:=O(N)/W(N)$.
\end{enumerate}
\end{proposition}
\begin{proof} See appendix \ref{a:proofUmbralSymm}.
\end{proof}

From the above theorems and proposition, we are led to the following corollary for the stringy $K3$ symmetries: 
\begin{corollary}\label{cor_in_phys_terms}
Consider type IIA string theory compactified on a $K3$ surface $X$ corresponding to a point in the moduli space ${\cal M}$ (cf. \eqref{moduli_space}).
\begin{enumerate}
\item
The supersymmetry-preserving discrete symmetry groups $G_{\rm IIA}$ that are realized somewhere in ${\cal M}$ are in bijection with the four-plane preserving subgroups of the Niemeier groups $G_N$. 
\item
Consider the sublattice of the D-brane lattice $H^\ast(X,\ZZ)\cong \Gamma^{4,20}$ orthogonal to the $G_{\rm IIA}$-invariant subspace of $H^\ast(X,\RR)$. The isomorphism classes of lattices that arise in this way somewhere in ${\cal M}$  are in bijection with the  isomorphism classes of co-invariant lattice $N_{\hat G}$, with $N$ a Niemeier lattice and $\hat G\subseteq G_N$ a four-plane preserving subgroup of the corresponding Niemeier group. 
\end{enumerate}
\end{corollary}

\subsection{$G$-Families}\label{s:Gfamilies}

It is useful to consider families of positive four-planes that share certain symmetries,
with the equivalence relation given by the action of $O^+(\Gamma^{4,20})$ taken into account. 
Let $G$ be a $O^+(\Gamma^{4,20})$-subgroup of four-plane fixing type.
 We define
\begin{align}\label{familydef}
\F_G:= \{ \Pi\subseteq \Gamma^G\otimes \RR,\ {\rm sign}(\Pi)=(4,0)\}/ \N_{O^+(\Gamma^{4,20})}(G),
\end{align}
where 
\be \N_{O^+(\Gamma^{4,20})}(G)=\{h\in  O^+(\Gamma^{4,20})\mid h Gh^{-1}=G\} \ee is the normalizer of $G$ inside $O^+(\Gamma^{4,20})$, and corresponds to the subgroup of $O^+(\Gamma^{4,20})$ that fixes the lattice $\Gamma^G$ setwise. Let $d$ be such that the invariant lattice $\Gamma^G$ has signature $(4,d)$. Then $\F_G$ is a Grassmannian parametrizing four-dimensional positive-definite subspaces within $\Gamma^G\otimes_\ZZ \RR\cong \RR^{4,d}$, modulo the group $\N_{O^+(\Gamma^{4,20})}(G)$ of physical dualities.
 The family  $\F_G$ admits a description as a  double coset
\be \F_G\cong SO(4)\times O(d)\backslash O^+(\Gamma^G\otimes_\ZZ \RR)/ \N_{O^+(\Gamma^{4,20})}(G)\ ,
\ee
which makes manifest that $\F_G$ has real dimension $4d$
and is connected: when $d=0$ the group $O^+(\Gamma^G\otimes_\ZZ \RR)\cong SO(4)$
has one connected component, while for $d>0$ the two connected components of $O^+(\Gamma^G\otimes_\ZZ \RR)\cong O^+(4,d)$ are identified via $O(d)$.  In the following, we will often identify  families $\F_G$ and $\F_{G'}$ that are conjugated in $O^+(\Gamma^{4,20})$:
 \be \F_G\sim \F_{G'} \quad{\rm if} \quad G'=hGh^{-1}\quad{\text{ for some}}\quad  h\in O^+(\Gamma^{4,20})\ .
 \ee

Physically, we are motivated to study $\F_G$ for the following reason. From the fact that a positive four-plane defines a $K3$ NLSM, $\F_G$ can be physically  interpreted as  a family of $K3$ NLSMs with symmetry groups which contain $G$. 
As we will see, the connectedness of $\F_G$ and the continuity argument we present in \S\ref{sec:conj} then guarantee that all theories ${\C}(\Pi)$ for $\Pi\in \F_G$ have the same twining genera $\eg_g({\C}(\Pi);\tau,z)$ for all $g\in G$.

We close this section with
a few useful properties of $\F_G$.  We would like to know whether a given family $\F_G$ contains any singular positive four-plane. 
First, let us distinguish between the following two cases: 
\begin{enumerate}[(1)]
\item $\Gamma_G$ contains no roots

 \item $\Gamma_G$ contains roots 

\end{enumerate} 
By definition, the case (1) contains some non-singular four planes, while the case (2) contains only singular models. It is natural to ask under what circumstances does a family in case (1) contain  singular four-planes. In what follows we collect the answer for a few interesting cases: 
\begin{itemize} 
\item if  $G$ is a group of \emph{geometric} symmetries (i.e., if $G$ arises as a group of hyper-K\"ahler preserving symmetries of a $K3$ surface), 
then the corresponding family $\F_G$ contains some singular models. To see this, first recall that a necessary and sufficient condition for $G$ to be geometric is that the invariant lattice 
$\Gamma^G$ contains an even unimodular $\Gamma^{1,1}\subset \Gamma^G$. 
In this case, one can take any root $v\in \Gamma^{1,1}$ and notice that $v^\perp \cap (\Gamma^G\otimes_\ZZ \RR)$ has signature $(4,d-1)$, so it contains some $\Pi$ of signature $(4,0)$ that is by definition singular. 

\item If $\Gamma^G$ has rank exactly four, then $\F_G$ consists of a single point, which is by definition non-singular.

\item If the defining $24$-dimensional representation of $G$ is not a permutation representation, then all four-planes in $\F_G$ are non-singular. This can be seen as follows. For each $\Pi$ in the family $\F_G$, one can show, using techniques analogous to the proof of Theorem \ref{t:invclass},  that the orthogonal sublattice $\Gamma_\Pi:= \Pi^\perp\cap \Gamma^{4,20}$ can be primitively embedded in some Niemeier lattice $N$ (possibly depending on $\Pi$). This implies that also $\Gamma_G\subset \Gamma_\Pi$ can be primitively embedded in $N$.  Recall that the defining 24-dimensional 
representation is a permutation representation for all subgroups of the Niemeier group $G_N$
unless $N$ is the Leech lattice. By hypothesis $\Gamma_G$ has no roots, so that by Proposition \ref{t:UmbralSymm} $G$ must be isomorphic to a subgroup of $G_N$. The only $N$ such that the $24$-dimensional representation of $G_N$ is not a permutation representation is the Leech lattice. We conclude that, for all $\Pi$ in $\F_G$, $\Gamma_\Pi$ can be embedded in the Leech lattice,  and therefore it cannot contain any root.

\end{itemize}
On the other hand, assuming a family in case (1) does contain a singular four-plane, we can deduce the following result about $\Gamma_G$:
\begin{itemize}
\item If $\F_G$ contains some singular four-plane $\Pi$, then $\Gamma_G$ can be embedded in some Niemeier lattice $N$ with roots,  so that $G$ is isomorphic to a subgroup of the Niemeier group $O(N)/W_N$. The argument for this is analogous to the previous statement.
The sublattice $\Gamma_\Pi:= \Pi^\perp\cap \Gamma^{4,20}$ orthogonal to a singular four-plane $\Pi$  can be primitively embedded in some Niemeier lattice $N$.  By definition, $\Gamma_\Pi$ contains some root and hence $N$ cannot be the Leech lattice. Furthermore, $\Gamma_G$ is a primitive sublattice of $\Gamma_\Pi$, so it can also be primitively embedded in $N$.\footnote{However, the converse is not true: it can happen that $\Gamma_G$ admits a primitive embedding into a Niemeier lattice while $\F_G$ contains no singular model, as exemplified by certain examples when $\Gamma^G$ is exactly four-dimensional and $\F_G$ contains only an isolated point. 
}
\end{itemize}

\section{Twining Genera}\label{sec:conj}

In this section we investigate how the symmetry groups discussed in the previous section act on the BPS spectrum of the theory. In particular, in \S \ref{s:conjectures} we will present two conjectures relating the twining genera of NLSMs and the functions featured in umbral and Conway moonshine. In this section we restrict our attention to non-singular NLSMs as the elliptic genus is otherwise not well-defined.

For any non-singular NLSM $\theory$ on $K3$, 
the elliptic genus may be defined as
\be\label{defEG} \eg(\theory;\tau,z)= \Tr_{{\cal H}_{RR}}(q^{L_0-\frac{c}{24}}\bar q^{\bar L_0-\frac{\bar c}{24}}y^{J_0} (-1)^{J_0+\bar J_0})
\ee where $q:=e^{2\pi i\tau},\ y:=e^{2\pi iz}$. 
In the above definition, ${\cal H}_{RR}$ denotes the Ramond-Ramond Hilbert space of $\C$, and  $L_0,\bar L_0$ and $J_0,\bar J_0$ denote the zero modes of the left- and right-moving Virasoro resp. the Cartan generators in the $su(2)$ level $1$ Kac-Moody algebra which are contained in the $\N=(4,4)$ superconformal algebra with central charges  $c=\bar c=6$.  
As is well-known, the elliptic genus of a compact theory only receives non-vanishing contributions from the right-moving ground states which have  vanishing eigenvalue of $\bar L_0-\frac{\bar c}{24}$, and hence it is holomorphic both in $\tau$ and in $z$.  Moreover, $ \eg(\theory;\tau,z)$ is a weak Jacobi form of weight zero and index $1$, {\it i.e.} it satisfies certain growth conditions \cite{eichler_zagier} and is a holomorphic function $\HH\times \CC\to \CC$ satisfying the following modularity 
\be
 \phi_{k,m}(\tau, z)= (c\tau + d)^{-k} e^{-2 \pi i m {c z^2\over c\tau + d}} \phi_{k,m} \left ({a \tau + b\over c\tau +d}, {z \over c\tau + d}\right ) ~~~ \forall 
\left(\begin{array}{cc}
a & b  \\
c & d  \end{array}\right) \in SL_2(\mathbb Z),
\ee
and quasi-periodicity properties
\be
 \phi_{k,m}(\tau,z)=e^{2 \pi i m(\ell^2\tau + 2\ell z)}\, \phi_{k,m}(\tau, z + \ell \tau + \ell') ~ ~~ \forall (\ell, \ell') \in \mathbb Z^2, 
\ee
for $k=0$ and $m=1.$
 The elliptic genus is a (refined) supersymmetric index and, in particular, is invariant under supersymmetric marginal deformations of the non-linear sigma model.\footnote{This is true except at points in moduli space where a non-compact direction opens up and the CFT is singular.}  Since the moduli space of K3 NLSMs is connected, this means that $\eg(\theory;\tau,z)$ is independent of the particular $K3$ NLSM $\theory$ from which it is calculated. As a result, often we will simply denote it as $\eg(K3;\tau,z)$. 
Explicitly, it can be expressed in terms of Jacobi theta functions as 
\be \eg(K3;\tau,z)=8\sum_{i=2}^4 \frac{\theta_i(\tau,z)^2}{\theta_i(\tau,0)^2}=2y+20+2y^{-1}+O(q). 
\ee

Let us consider a non-singular NLSM $\theory$ with a symmetry group $G$. Then, for each $g\in G$, one can define the twining genus
\be\label{twiningdef} \eg_g(\theory;\tau,z)=\Tr_{{\cal H}_{RR}}(gq^{L_0-\frac{c}{24}}\bar q^{\bar L_0-\frac{\bar c}{24}}y^{J_0} (-1)^{F_L+F_R})\ .
\ee 
From the usual path integral picture, one concludes that $\eg_g$ is a weak Jacobi form of weight $0$ and index $1$ for some congruence subgroup $\modgrp$ of $SL_2(\ZZ)$, possibly with a non-trivial multiplier system (see appendix \ref{a:modularity} for details). 

By the same arguments as for the elliptic genus and under standard assumptions about deformations of $\cN=(4,4)$ superconformal field theories, the twining genus $\eg_g$ is invariant under exactly marginal deformations that preserve supersymmetry and the symmetry generated by $g$. 
More precisely, consider a group of symmetries $G$ such that the subspace 
\be 
\label{ns_fam}\F^{ns}_{G}:=\{\Pi \subseteq \F_G \mid \Pi ~{\text{ is not  singular}}\}\ 
\ee of non-singular positive four-planes is non-empty (cf.\eqref{familydef}). 
Note that there is no loss of generality by restricting to non-singular models, since only for these the world-sheet definition of (twined) elliptic genus 
that we employ in this section applies. 
Then we argue that the following is true: 

{\it Let $g\in O^+(\Gamma^{4,20})$ be a group element fixing pointwise a sublattice $\Gamma^g\subseteq \Gamma^{4,20}$ of signature $(4,d)$ and such that the co-invariant lattice $\Gamma_g 
$ 
contains no roots. Then, the family $\F_g^{ns}:=\F_{\langle g\rangle}^{ns}$ of non-singular four-planes with symmetry $g$ is non-empty and connected. Furthermore, if we assume that the operators $L_0,\bar L_0,J_0,\bar J_0$ and $g$ vary continuously under deformations within the family of NLSM corresponding to $\F_g^{ns}$, then the twining genus $\eg_g$ is constant on $\F_g^{ns}$.
}

The proof is an obvious generalization of the arguments showing that the elliptic genus is independent of the moduli. One first defines the twining genus $\eg_g$ along any connected path within the family $\F^{ns}_g$, and then uses continuity of $L_0,\bar L_0,J_0,\bar J_0$ as well as the discreteness of their spectrum within the relevant space of states to show that $\eg_g$ must be actually constant along this path. An even simpler proof can be given if one adopts the equivalent definition of the twining genus $\eg_g$ as an equivariant index in the $Q$-cohomology of a half-twisted topological model. In this case, it is sufficient to use the fact that a  $g$-invariant and $Q$-exact deformation cannot change the index.

We note that \be\label{chargeinv}\eg_g(\tau,z)=\eg_{g^{-1}}(\tau,-z)=\eg_{g^{-1}}(\tau,z)\ .\ee
Here, the first equality corresponds to the transformation $\left(\begin{smallmatrix}
-1 & 0\\ 0 & -1
\end{smallmatrix}\right)\in SL_2(\ZZ)$ and follows from standard path integral arguments. The second equality holds because the spectrum is a representation of the $su(2)$ algebra contained in the left-moving $\N=4$ algebra,  and $su(2)$  characters are always even.

 Finally, a twining genus $\eg_g$ is invariant under conjugation by any duality $h\in O^+(\Gamma^{4,20})$. More precisely, suppose $h$ is a duality between the models $\C$ and $\C'$, i.e. an isomorphism between the fields and the states of the two theories that maps the superconformal generators into each other and is compatible with the OPE. Then,  the twining genus $\eg_g$ defined in the model $\C$ equals the twining genus $\eg_{hgh^{-1}}$ defined in the model $\C'$. This follows immediately using the cyclic properties of the trace.   The effect of a conjugation under a duality in $O(\Gamma^{4,20})\setminus O^+(\Gamma^{4,20})$ is much more subtle and will be discussed in section \ref{s:parity}.

Using the above results, 
one can assign a twining genus $\eg_g$ to any  conjugacy class $[g]$ of $O^+(\Gamma^{4,20})$ such that $\langle g \rangle$ is a subgroup of four-plane fixing type and that the co-invariant sublattice $\Gamma_g$ contains no roots. In principle, $\eg_g$ and $\eg_{g'}$ are distinct if $g'$ is conjugate to neither $g$ nor $g^{-1}$ as elements of $O^+(\Gamma^{4,20})$, unless accidental coincidences occur.\footnote{Coincidences like this occur, for example, when the dimension of the relevant space of modular forms is small. See sections \S \ref{s:twin_from_modul} and \S\ref{s:evidence} for more details.} In the next subsection we will classify the conjugacy classes of $O^+(\Gamma^{4,20})$.

\subsection{Classification}\label{s:distinctgenera}

While many examples of twining genera have been computed in specific sigma models, a full classification of the corresponding conjugacy classes in $O^+(\Gamma^{4,20})$ and a complete list of all corresponding twining genera is still an open problem. In this work we solve the first problem for all but one of the forty-two possibilities (labelled by conjugacy classes of $Co_0$). 

As a first step in this classification problem, it is useful to consider the eigenvalues of $O^+(\Gamma^{4,20})$-elements in the defining $24$-dimensional representation, denoted below simply by $\rho_{24}:  O^+(\Gamma^{4,20}) \to {\rm End}(V_{24})$, given by $V_{24}\cong \Gamma^{4,20}\otimes_\ZZ \RR$. (This is also the representation on the $24$ R-R ground states in a sigma model $\theory(\Pi)\in \F_{g}$.) It is convenient to encode such information in the form of a   Frame shape, {\it i.e.} a symbol 
\be \pi_g := \prod_{\ell|N}\ell^{k_\ell}\ , \ee
where $N=o(g)$ is the order of $g$. The integers $k_\ell\in \ZZ$ are defined by
\be \det (t{\bf 1}_{24}- \rho_{ 24}(g)) =\prod_{\ell|N} (t^\ell-1)^{k_\ell}\ .
\ee When $g$ acts as a permutation of vectors in $\Gamma^{4,20}\otimes_\ZZ \RR$, all $k_\ell$ are non-negative and the Frame shape coincides with the cycle shape of the permutation. We will say that a Frame shape is a {\it four-plane preserving Frame shape} if it coincides with the Frame shape of an element of a four-plane preserving subgroup of $O(\Gamma^{4,20})$, as defined in \S\ref{sec:class}. Explicitly, a Frame shape is a four-plane preserving Frame shape if and only if $\sum_{\ell} k_{\ell} \geq 4$, corresponding to the fact that the eigenvalue $1$ must be repeated at least four times.  
A salient feature shared by the Frame shapes of all Niemeier groups that correspond to Niemeier lattices with non-trivial root systems (and hence not given by the Leech lattice) is that they are all cycle shapes, and this is not true for some of the Conway Frame shapes. 

One can explicitly check, by using Theorem \ref{t:dirclass}, that such four-plane preserving Frame shapes of $O^+(\Gamma^{4,20})$ are precisely the 42 four-plane preserving Frame shapes of $\Co_0$, corresponding to the 42 four-plane preserving conjugacy classes of $\Co_0$. 
Moreover, if $g,g'\in O^+(\Gamma^{4,20})$ have the same Frame shape, then the co-invariant sublattices $\Gamma_g$ and $\Gamma_{g'}$ are isomorphic
\be \pi_g=\pi_{g'} \qquad \Rightarrow \qquad \Gamma_g\cong \Gamma_{g'}\ .
\ee
This follows from the fact that  $\Gamma_g\cong \Lambda_{\hat g}$ and $\Gamma_{g'}\cong\Lambda_{\hat g'}$ by construction, and moreover $\hat g$ and $\hat g'$ are conjugated in $\Co_0$. 
However, it can happen that $\Gamma_g,\Gamma_{g'}\subset \Gamma^{4,20}$ are isomorphic as abstract sublattices, but are not conjugated within $O^+(\Gamma^{4,20})$.
Indeed, as argued in detail in appendix \ref{a:classifclasses}, the problem of determining the number of $O^+(\Gamma^{4,20})$ conjugacy classes for a given Frame shape can be reduced to that of classifying the (in a suitable sense) inequivalent primitive embeddings of the corresponding lattice $\Lambda_{\hat g}$ in $\Gamma^{4,20}$. The proof of this statement can be found in appendix \ref{a:classifclasses} and the result of the classification is tabulated in appendix \ref{a:results}.

A summary of these results is the following.  Out of the 42 distinct four-plane 
preserving Frame shapes of $O^+(\Gamma^{4,20})$, there was only one (with Frame shape $1^{-4}2^53^46^1$) for which we were unable to determine the number of its $O(\Gamma^{4,20})$ (and thus $O^+(\Gamma^{4,20})$) classes. For this Frame shape, we can only prove that  either a) there is  one class or b) there are two classes for which one is the inverse of the other. The remaining 41 $Co_0$ conjugacy classes give rise to 58 distinct $O(\Gamma^{4,20})$ conjugacy classes and 80 distinct $O^+(\Gamma^{4,20})$ conjugacy classes.

\subsection{World-Sheet Parity}\label{s:parity}

We have argued earlier that the twining genera $\eg_g$ are invariant under conjugation by $O^+(\Gamma^{4,20})$ dualities.
In many physical applications, however, the larger group  $O(\Gamma^{4,20})$ is taken to be the relevant duality group. 
Indeed, the elliptic genus is obviously the same for  two theories related by any element of $O(\Gamma^{4,20})$. 
In this subsection we will show that the twining genera $\eg_g$, on the other hand, are in general different unless the two theories are related by an element of   $O^+(\Gamma^{4,20})\subset O(\Gamma^{4,20})$. 

To understand this, first note that only elements of $O^+(\Gamma^{4,20})\subset O(\Gamma^{4,20})$, which by definition preserve the orientation of any positive four-plane, preserve the orientation of the world-sheet of NLSM \cite{Nahm:1999ps}.  This can be understood as follows. 
The group $SO(4)$ of rotations of a $4$-plane $\Pi\subset \RR^{4,20}$ acts on the  $80$-dimensional space of exactly marginal operators of the  corresponding NLSM $\theory(\Pi)$. 
 The latter have the form $G_{-1/2}\overline{G}_{-1/2}\chi_i$, where $\chi_i$, $i=1,\ldots,20$, are fields of weight $(1/2,1/2)$. 
 Since they preserve the $\N=(4,4)$ algebra, the marginal operators are singlets under the internal  holomorphic and anti-holomorphic R-symmetries $SU(2)^{\it{susy}}_L\times SU(2)_R^{\it{susy}}$.
 On the other hand, they transform as $({\bf 2},{\bf 2})$ under the group $SU(2)_L^{\it{out}}\times SU(2)_R^{\it{out}}$ of outer automorphisms acting on the left- and right-moving supercharges. 
 In other words, $SU(2)_L^{\it{out}}\times SU(2)_R^{\it{out}}$ can be identified with the double cover of the group $SO(4)$ acting on $\Pi$. 
 As a result,  $h\in O(\Gamma^{4,20})$ exchanges $SU(2)_L^{\it{out}}$ and $SU(2)_R^{\it{out}}$ if and only if it flips the orientation of $\Pi$.
 Since the definition of the twining genera effectively only focuses on the action of $g$ on the left-movers ({\it i.e.} on the right-moving ground states), and in general $g$ acts on the right-movers differently, one expects $\eg_g$ to be invariant only under $O^+(\Gamma^{4,20})$ duality transformations. 

This consideration is particularly relevant for symmetries whose corresponding twining genera have complex multiplier systems. Recall that the twined elliptic genus $\eg_g$ is a Jacobi form under a certain congruence subgroup $\modgrp\subset \SL_2(\ZZ)$ with a (in general non-trivial) multiplier $\psi_g: \modgrp\to \CC^\ast$. 
We say that  $\psi_g$  is a {\em complex multiplier system} if its image does not lie in $\RR$. Note that this is necessarily the case when the multiplier has order greater than $2$. 
To see the relation between world-sheet parity and the multiplier system, consider two $K3$ NLSMs $\theory$ and $\theory'$ corresponding to the four-planes $\Pi$ and $\Pi'$ that are related by an $h\in O(\Gamma^{4,20})$, $\Pi'=h(\Pi)$, which reverses the orientation of a positive four-plane and hence exchanges the left- and the right-movers.    
 This means in particular that $h$ maps 
 the $\N=(4,4)$ algebras of $\theory$ and   $\theory'$ as
\be hL_nh^{-1}=\bar L'_n\qquad  hJ_nh^{-1}=\bar J'_n\ . 
\ee
 Given a symmetry $g$ of $\theory$, namely $g\in O^+(\Gamma^{4,20})$ such that $g$ fixes $\Pi$ pointwise, then a corresponding symmetry of $\theory'$ is given by $g':=hgh^{-1}$.
We would like to know whether $\eg_g(\theory;\tau,z)$ and $\eg_{g'}(\theory';\tau,z)$ are the same. 

To answer this question, 
consider the refined twining partition function
\be Z_g(\theory;\tau, z,\bar u)=\Tr_{{\cal H}_{RR}}\left(g q^{L_0-\frac{c}{24}}\bar q^{\bar L_0-\frac{\bar c}{24}}e^{2\pi i z J_0}e^{-2\pi i \bar u \bar J_0}(-1)^{F+\bar F}\right)\ ,
\ee 
for a symmetry $\langle g \rangle$ of the theory $\C$. 
Note that, unlike the elliptic genus, this function is  not an index, and  it depends on both the conjugacy class of $g$ and the point in moduli space, $\theory$. In general, $Z_g$ is not holomorphic in $\tau$, but it is elliptic (one can apply  spectral flow independently to the left- and right-movers) and modular (in the appropriate sense for a non-holomorphic Jacobi form) under some subgroup of $\SL_2(\ZZ)$. 
In particular, if $g$ has order $N$, we expect $Z_g$ to transform under $\left(\begin{smallmatrix}
a & b\\ c & d\end{smallmatrix}\right)\in \Gamma_1(N)$ as
\be 
 Z_g(\theory;\tau,z,u)= \psi_g((\begin{smallmatrix}
a & b\\ c & d
\end{smallmatrix}))\,  e^{-2\pi i ( \frac{cz^2}{c\tau+d}+ \frac{c\bar u^2}{c\bar \tau+d} )} \,Z_g\left(\theory;\frac{a\tau+b}{c\tau+d} ,\frac{z}{c\tau+d},\frac{\bar u}{c\bar \tau+d}\right) \ \qquad \qquad \ .
\ee  
Clearly, one recovers the twining genus as 
$$Z_g(\theory;\tau, z,\bar u=0)=\eg_g(\theory;\tau,z).$$ This implies that the multiplier $\psi_g$ of the twining partition function $Z_g$ coincides with the one of the twining genus $\eg_g$.

Now, the $O(\Gamma^{4,20})$-equivalence and the absence of $O^+(\Gamma^{4,20})$-equivalence between the theories $\theory$ and $\theory'$ implies
\be
Z_{g'}(\theory';\tau',z',\bar u')= Z_g(\theory;\tau,z,\bar u). 
\ee
In the above, apart from $g'=hgh^{-1}$ we also have $\tau'=-\bar \tau$, $z'=-\bar u$ and $\bar u'=-z$.  
To see the relation between the multiplier system of $\eg_{g'}(\theory')$ and $\psi_g$, note that the above equation implies
\be
Z_{g'}(\theory';\tau,0,z) = Z_g(\theory;-\bar\tau,-\bar z,0) =\eg_g(\theory; -\bar \tau, -\bar z) . 
\ee
As a result, assuming that the coefficients of the double series expansions in $q$ and $y$ of $\eg_g(\theory)$ are all real, we obtain 
\be
Z_{g'}(\theory';\tau,0,z)  = \overline{\eg_g(\theory;\tau,z)}
\ee
and hence has  multiplier given by $\overline{\psi_g}:\Gamma_1(N) \to \CC^\ast$,  the inverse of the multiplier of $\eg_g$. 
The above assumption can be proven from the fact that $\overline{\text{Tr}_V (g)}=\text{Tr}_V (g^{-1})$ for any finite-dimensional representation $V$ of a finite group $\langle g\rangle$ and using the identity $\eg_g=\eg_{g^{-1}}$ (see eq.\eqref{chargeinv}). 

Finally, recall that $Z_{g'}(\theory';\tau,0,u)$ and $Z_{g'}(\theory';\tau,z,0)=\eg_{g'}(\theory';\tau,z)$ necessarily have the same multiplier, since they both  coincide with that of $Z_{g'}(\theory';\tau,u,z)$, and thus we conclude that the twining genera $\eg_{g'}(\theory';\tau,z)$ and $\eg_{g}(\theory;\tau,z)$ have multiplier systems that are the inverse (equivalently, complex conjugate) of each other. In particular, $\eg_{g'}(\theory';\tau,z)$ and $\eg_{g}(\theory;\tau,z)$ cannot be the same unless $\psi_g = \overline{\psi_g}$. As a result, symmetries $g$ leading to a twining genus with a complex multiplier system necessarily act differently on left- and right-moving states. 
Note however that it can happen that  a symmetry acting asymmetrically on left- and right-movers leads to a twining genus with a multiplier system of order one or two. 
In what follows we will refer to a symmetry $g$ of a NLSM a {\it complex symmetry} if the resulting twining genus has complex multiplier system.

\subsection{Conway and Umbral Moonshine}\label{s:conjectures}

Once the possible $O^+(\Gamma^{4,20})$ classes of symmetries have been determined, it remains to calculate the corresponding twining genera. As we will see in \S\ref{sec:alltwinings}, many  examples have been computed in specific NLSMs. However, the list of such functions is still incomplete. 
After reviewing the earlier work \cite{Cheng:2014zpa,duncan2016derived}, 
in this subsection we present two conjectures relating physical twining genera to functions arising from umbral and Conway moonshine, as well as some evidence for their validity. 

 Consider the 23 Niemeier lattices $N$ with non-trivial root systems. 
Umbral moonshine attaches to each element $g$ of the Niemeier group $G_N$ a weight one mock Jacobi form 
$$\Psi^N_g(\tau,z)=\sum_{r\in \ZZ/2m} H_{g;r}^N(\tau) \theta_{m,r}(\tau,z),$$ whose index is given by the Coxeter number of the root system of the corresponding Niemeier lattice $N$ \cite{Cheng:2013wca}.
In the above expression, the index $m$ theta functions are given by
$$
\theta_{m,r}(\tau,z) = \sum_{k=r\!\!\!\!\mod{2m}}q^{k^2/4m} y^{k},
$$
and the vector-valued mock modular form $H_{g}^N=(H_{g;r}^N)$ contains precisely the same information as the  mock Jacobi form $\Psi^N_g$.
 In \cite{Cheng:2014zpa}, a weight 0 index 1 Jacobi form for a certain $\modgrp\subseteq \SL_2(\ZZ)$,  is then given in terms of $\Psi^N_g$ by \be\label{um_twining}
\phi_g^N(\tau,z) = \eg(N;\tau,z) + {\theta_1^2(\tau,z) \over \eta^6(\tau)} \left({1\over 2\pi i}{\partial\over \partial w} \Psi^N_g(\tau,w)\right) \Big\lvert_{w=0}. 
\ee 
In the above formula, $\eg(N;\tau,z)$ denotes the holomorphic part of the elliptic genus of the singularities corresponding to the root system of $N$.   Recall that the only type of geometric singularities a $K3$ surface may develop are du Val surface singularities, i.e. singularities of the complex plane of the form $\mathbb C^2/G$, where $G$ is a finite subgroup of $SU(2)_{\mathbb C}$. These singularities have an ADE classification, formally analogous to the one of simply-laced root systems. A conformal field theory description of string theory with ADE singularities as the target space was given in \cite{Ooguri:1995wj}. The form of their elliptic genus was investigated in a number of papers, including \cite{Troost:2010ud,Eguchi:2010cb,Ashok:2011cy,Ashok:2013pya,Murthy:2013mya,Cheng:2014zpa,Harvey:2014nha}. For instance, when $N$ is the Niemeier lattice with root system $24A_1$,  $\eg(N;\tau,z) := 24\eg(A_1;\tau,z)$ is 24 times the holomorphic part of the elliptic genus of an $A_1$-singularity. 
 
 It was conjectured in  \cite{Cheng:2014zpa} that $\phi_g^N$ are candidates for twining genera arising from $K3$ NLSMs when $g$ preserves a four-plane; this conjecture has passed a few consistency tests and was further tested in \cite{cheng2015landau}. For a given $N$ with a non-trivial root system, we will denote the set of Jacobi forms arising in this way as $$\Phi(N):= \{\phi_g^N\lvert g~{\text{is a four-plane preserving element of } G_N}\}.$$ 
 
The construction and conjecture in \cite{Cheng:2014zpa} gives us a set of Jacobi forms $\Phi(N)$ attached to  each of the 23 Niemeier lattice $N$ with roots that (conjecturally) play the role of twined $K3$ elliptic genera at certain points in the moduli space. 
It is also possible to define a similar set $\Phi(\Leech)$ associated with the Leech lattice $\Leech$, though the construction is quite different. In \cite{duncan2016derived} Duncan and Mack-Crane proposed two (possibly coinciding) weight 0 index 1 weak Jacobi forms for a certain $\modgrp\subseteq \SL_2(\ZZ)$, denoted $\phi_{g,+}^\Leech(\tau,z)$ and $ \phi_{g,-}^\Leech(\tau,z)$, to each of the four-plane preserving conjugacy classes of $\Co_0$. Concretely, one has 
$$
\phi_{g,\pm}^\Leech(\tau,z) = \sum_{i=1}^4  \epsilon_{g,i} \, \theta_i^2(\tau,z) \prod_{k=1}^{10}  \theta_i^2(\tau,\rho_{g,k})
$$
where $$\{1,1,1,1, e^{- 2\pi i\rho_1}, e^{ 2\pi i \rho_1},\dots ,e^{- 2\pi i \rho_{10}}, e^{ 2\pi i\rho_{10} } \}$$ are the twenty-four eigenvalues of $g$ acting on the 24-dimensional representation, and 
$$
\epsilon_{g,i} = \begin{cases}\mp 1& i=1\\ -{\Tr_{\bf 4096} g \over 4\prod^{10}_{k=1} ( e^{- \pi i\rho_k}+e^{ \pi i\rho_k})}&i=2 \\ 1 & i=3\\ -1
&i=4 \end{cases}.
$$ In the above formula, $\bf 4096$ is the Conway representation corresponding to the fermionic ground states and decomposes as ${\bf 4096 = 1 + 276 + 1771 + 24 + 2024}$ in terms of irreducible representations.

 One has $\phi_{g,+}^\Leech \neq \phi_{g,-}^\Leech$ if and only if the invariant sublattice $\Leech^g$ has exactly rank four. 
 The construction of $\phi_{g,+}^\Leech$ and $\phi_{g,-}^\Leech$ is based on an $\cN=1$ super VOA of central charge $c=12$, which has symmetry group $\Co_0$ \cite{Duncan:2014eha}. Henceforth we define 
 $$\Phi(\Leech):= \{\phi_{g,+}^\Leech , \phi_{g,-}^\Leech \lvert g~{\text{is a four-plane preserving element of } \Co_0}\}.$$
 The authors of \cite{duncan2016derived} then conjectured that the functions in $\Phi(\Leech)$ are  relevant for twining genera arising from (non-singular) $K3$ NLSMs. In fact, they  conjecture that all twining genera arising from any $K3$ NLSM coincide with some element of $\Phi(\Leech)$ arising from the Conway module, which is supported by the  non-trivial fact that all the {\it known} twining genera $\eg_g$ coincide with a function in $\Phi(\Leech)$.

 There are a few motivations for us to modify this conjecture and to make the conjecture in \cite{Cheng:2014zpa} more concrete. Firstly,   the classification theorems of \S\ref{sec:class} suggest that, if one does not exclude the loci in the moduli space (\ref{moduli_space}) corresponding to singular four-planes, one should treat the Leech lattice and the other 23 Niemeier lattices with non-trivial root systems on an equal footing when discussing the four-plane preserving symmetry groups. As a consequence, one might expect both Conway and umbral moonshine to play a role in describing the twining genera. Secondly,  UV descriptions of K3 NLSMs given by Landau--Ginzburg (LG) orbifolds furnish evidence that suggests that the Conway functions  alone are not sufficient to capture all the twining genera \cite{cheng2015landau} (see also section \S\ref{LGtwin}.) To be more precise, there are twining genera arising from symmetries of  UV theories that flow to $K3$ NLSMs  in the IR, that can be reproduced from the set $\Phi(N)$ for some $N$ with roots, but do not coincide with anything in $\Phi(\Leech)$. One caveat preventing this result from being a definitive argument is that the action of the corresponding  symmetry on the IR $\mathcal N=(4,4)$ superconformal algebra is not accessible in the UV analysis. 
 
 The third and arguably  most convincing argument to include functions arising from both  Conway and umbral moonshine is the following. As we have seen in \S\ref{s:parity}, a pair of theories related by a flip of world-sheet parity gives rise to twining genera with inverse multiplier systems. At the same time, $\Phi(\Leech)$ contains some twining functions with a complex multiplier system and no functions with the inverse multiplier. Such functions can always be recovered from $\Phi(N)$ for some other Niemeier lattice $N$. As a result, no single $\Phi(N)$ (not even for $N$ the Leech lattice) is sufficient  to reproduce both a physical twining function $\eg_g(\theory)$ with complex multiplier and its parity-flipped counterpart $\eg_{g'}(\theory')$.

These observations lead us to formulate the following conjecture:

 \begin{conjecture}\label{conj:all_arise_from_moonshine} 
Let $\theory(\Pi)$ be a $K3$ NLSM and let $G$ be its symmetry group. 
Then there exists at least one Niemeier  lattice $N$  such that $\Gamma_G$ can be embedded in $N$, $G \subseteq G_N$, and for any $g\in G$ 
 the twining genus $\eg_g$ coincides with an element of $\Phi(N)$. 
 \end{conjecture}
In other words,
 we conjecture that for each $K3$ NLSM $\theory$,  the set $\tilde\Phi(\theory):=\{\eg_g(\theory(\Pi))\lvert  g\in O^+(\Gamma^{4,20}), ~g~\text{fixes $\Pi$ pointwise}\}$ of physical twining genera  is a subset of the $\Phi(N)$ for some Niemeier lattice $N$. 
 Clearly, for most theories,  the Niemeier lattice $N$ satisfying the above properties is not  unique. 
 In particular, recall that there are many coincidences among the functions associated with different Niemeier lattices. In other words, there exist $\phi \in \Phi(N)$, $\phi' \in \Phi(N')$ with $N\neq N'$ such that $\phi=\phi'$.

Conversely, we conjecture that {\it all} elements of $\Phi(N)$  play a role in capturing the symmetries of BPS states of $K3$ NLSMs:  

\begin{conjecture}\label{c:manytwin} 
For any element $\phi$ of any of the 24 $\Phi(N)$,  there exists a NLSM $\theory$ with a symmetry $g$ such that $\phi=\eg_g(\theory)$. 

\end{conjecture}

In  \S\ref{s:evidence} we collect some evidence supporting these conjectures.
We will close this section with a few remarks on the consequences of the above conjectures, in relation to the complex symmetries discussed in \S\ref{s:parity}.
\begin{itemize}
\item If a given function in $\Phi(N)$ has complex multiplier system, then Conjecture \ref{c:manytwin} implies that it has to coincide with a twining genus arising from a complex symmetry acting differently on the left- and right-moving Hilbert spaces. 
\item As we argued in \S\ref{s:parity}, if a theory $\theory$  leads to the twining function $\eg_g(\theory)$ with a complex multiplier system, the parity-flipped theory $\theory'$ has a twining genus $\eg_{g'}(\theory')$ with the inverse multiplier system. As a result, the following observations constitute consistency checks and circumstantial evidence for Conjecture \ref{conj:all_arise_from_moonshine} and Conjecture \ref{c:manytwin}. 
Namely, whenever there exists a Niemeier lattice $N$ and a function $\phi\in \Phi(N)$ with a complex multiplier system, arising from a group element with a given Frame shape $\pi$, then there exists at least one other Niemeier lattice $N'$ such that there exists a $\phi'\in \Phi(N')$ with the inverse complex multiplier system, which moreover arises from a group element with the same Frame shape $\pi$.  
See Table \ref{tbl:mult} for the pairs $(N',g')$ with the above properties. 

\item In fact, by inspection one can check that there are never two functions $\phi,\phi'\in \Phi(N)$ arising from the same Niemeier lattice that have inverse complex multiplier systems. As a result, Conjecture \ref{conj:all_arise_from_moonshine} predicts that a theory corresponding to the four-plane $\Pi$ must have its orthogonal sublattice $\Gamma^{4,20}\cap \Pi^\perp$ embeddable into more than one Niemeier lattice in the event that it has a complex symmetry. 

\item Recall that a theory in the NLSM moduli space (\ref{moduli_space}) on a torus orbifold locus--one of the few types of exactly solvable models--always contains symmetries which can only be embedded using the Leech lattice (in the sense of Theorem \ref{t:dirclass}) \cite{Gaberdiel:2012um}.  As a result, assuming the veracity of Conjecture \ref{conj:all_arise_from_moonshine},  complex symmetries can never arise in such a model. This makes it particularly difficult to find examples of $K3$ NLSMs with complex symmetries and 
probably explains why we have seen no such examples so far. In  \S\ref{LGtwin} we will discuss results of the aforementioned investigation of LG orbifolds \cite{cheng2015landau}, while in \S\ref{s:twin_from_modul} we will analyze the constraints on such genera coming from modularity.
\end{itemize}

\section{Examples}\label{sec:alltwinings}

In this section, we collect all known explicit calculations of twining genera in NLSMs on $K3$. Most of these results have appeared earlier in the literature, the only exceptions being certain genera appearing in \S\ref{T4orbifolds} and \S\ref{s:twin_from_modul}. See Table \ref{last_table_genera} for the data.  While these examples do not cover the complete set of all possible twining genera, the fact that these partial results fit nicely with the general properties described in the previous sections represents strong evidence in favor of our conjectures.

\subsection{Geometric Symmetries}\label{geomtwin}

We say that a symmetry of a $K3$ NLSM is a geometric symmetry if it is induced from a hyper-K\"ahler preserving automorphism of the target $K3$ surface. 
 These symmetries are characterized by the property that the fixed sublattice $\Gamma^g$ contains a unimodular $\Gamma^{1,1}$, which can be interpreted as the components $H_0(S,\ZZ)\oplus H_4(S,\ZZ)$ of degree $0$ and $4$ in the integral homology of the  $K3$ surface $S$. 
 There exist such geometric symmetries with order $N\in\{2,3,4,5,6,7,8\}$. For each of these orders
 there is precisely one Frame shape
\be 1^82^8,\quad 1^63^6,\quad 1^42^44^4,\quad 1^45^4,\quad 1^22^23^26^2,\quad 1^37^3,\quad 1^22^14^18^2\ \ee
   that can arise from a geometric symmetry of a $K3$ surface \cite{mukai1988finite}. 
  
 A general formula for the corresponding twining genus for each of the above Frame shapes has been given in \cite{Cheng:2010pq} and
  \cite{Sen:2010ts} and reads 
\be \eg_g(\tau,z)= \frac{\Tr_{\bf 24}(g)}{\varphi(N)}\sum_{\substack{n\in \ZZ/N\ZZ \\ gcd(n,N) =1 }}\frac{\vartheta_1(\tau,z+\frac{n}{N})\vartheta_1(\tau,z-\frac{n}{N})}{\vartheta_1(\tau,\frac{n}{N})\vartheta_1(\tau,-\frac{n}{N})}\ ,
\ee where the totient function $\varphi(N):=|(\ZZ/N\ZZ)^\times|$ is number of integers mod $N$ that are coprime to $N$. These twining genera can be defined in purely geometric terms as an equivariant complex elliptic genus and computed using a version of the Lefschetz fixed point formula \cite{Creutzig:2013mqa}. The results agree with the formulas derived from NLSMs.

\subsection{Torus Orbifolds}\label{T4orbifolds}

If a $K3$ model is obtained as a (possibly asymmetric) orbifold of a torus $T^4$ by a symmetry $g$ of order $N$, then it has a quantum symmetry $Q$ of order $N$, which acts as multiplication by $e^{\frac{2\pi i r}{N}}$ on all states in the $g^r$-twisted sector, $r\in \ZZ/N\ZZ$.
It is not difficult to compute the twining genus of a quantum symmetry, since it can be computed from the twining genus of $g$ on the NLSM $\theory_{T^4}$  on the $T^4$. In \cite{Volpato:2014zla}, general formulas for the twining genera of all possible symmetries of any NLSM on $T^4$ were given. The supersymmetric NLSM on $T^4$ has four left-moving and four right-moving Majorana-Weyl fermions. The holomorphic fermions form two doublets $(\chi_i^+,\chi_i^-)$, $i=1,2$, each in the $ ({\bf 2},{\bf 1})$ representation of the $SU(2)_L\times SU(2)_R$ R-symmetry , while the anti-holomorphic fermions form doublets $(\tilde\chi_i^+,\tilde\chi_i^-)$, $i=1,2$, in the $ ({\bf 1},{\bf 2})$ representation. The symmetry $g$ commutes with the R-symmetry and acts on the multiplets by
\begin{align}
(\chi_1^+,\chi_1^-)\mapsto& \zeta_L(\chi_1^+,\chi_1^-) & (\chi_2^+,\chi_2^-)\mapsto& \zeta_L^{-1}(\chi_2^+,\chi_2^-)\\
(\tilde\chi_1^+,\tilde\chi_1^-)\mapsto& \zeta_R(\tilde\chi_1^+,\tilde\chi_1^-) & (\tilde\chi_2^+,\tilde\chi_2^-)\mapsto& \zeta_R^{-1}(\tilde\chi_2^+,\tilde\chi_2^-)\ ,
\end{align} with  
\be \zeta_{L,R}=\exp(2\pi i r_{L,R}) \quad \text{for some}\quad r_L,r_R\in \frac{1}{N}\ZZ/\ZZ\ .
\ee For $\zeta_L\neq 1$, the twining genus of $g$ is given by
\be {\cal Z}_g(\theory_{T^4};\tau,z)=(\zeta_L+\zeta_L^{-1}-2)(\zeta_R+\zeta_R^{-1}-2)\frac{\vartheta_1(\tau,z+r_L)\vartheta_1(\tau,z-r_L)}{\vartheta_1(\tau,r_L)\vartheta_1(\tau,r_L)}
.\ee
Note that the above function is invariant under both $r_L\to -r_L$ and $r_R\to -r_R$, but is in general not invariant under $r_L\leftrightarrow r_R$. When  $r_L=0\mod \ZZ$ (i.e. $\zeta_L=1)$, the twining genus is given instead  by
\be\label{alternative} {\cal Z}_g(\theory_{T^4};\tau,z)=(\zeta_R+\zeta_R^{-1}-2)\frac{\vartheta_1(\tau,z)^2}{\eta(\tau)^6}\Theta_L(\tau)\ ,
\ee where \be \Theta_L(\tau):=\sum_{\lambda\in L} q^{\frac{\lambda^2}{2}}\  \ee is the theta series associated with a lattice $L$ of rank $4$.  The only relevant cases are $(r_L,r_R)=(0,1/2)$ and $(r_L,r_R)=(0,1/3)$, in which cases $L$ is the $D_4$ or $A_2^2$ root lattices respectively (see \cite{Volpato:2014zla} for more details). In particular, the untwined elliptic genus of $T^4$ is ${\cal Z}_e(T^4;\tau,z)=0$.

When a CFT has a discrete symmetry, it is also useful to discuss the twisted sectors of the symmetry (modules of the invariant sub-algebra), labelled by the twisting group element $g$. 
For any element $h$ of the  discrete symmetry group that commutes with the twisting element $g$, one can consider the graded trace of $h$ over the $g$-twisted sector, analogous to the way in which
a twined partition function or twined elliptic genus is defined.  
Such a character is often called the  twisted-twining partition function/elliptic genus. As usual in the literature, we use ${\cal Z}_{h,g}$ to denote the $g$-twining function in the $h$-twisted sector. In particular, 
the twining function of the original unorbifolded theory is given by $ {\cal Z}_{g} :=  {\cal Z}_{e,g}$.  

Using the modular properties of the theta function as well as the fact that the twisted-twining genera  form a representation of $\SL_2(\ZZ)$, we obtain the following expression which is valid for $r_LM\neq 0\mod \ZZ$
\begin{align}\notag {\cal Z}_{g^n,g^m}(\theory_{T^4};\tau,z)=&(\zeta_L^M+\zeta_L^{-M}-2)(\zeta_R^{M}+\zeta_R^{-M}-2)\\\label{T4twin}& \times\frac{\vartheta_1(\tau,z+r_L(n\tau+m))\vartheta_1(\tau,z-r_L(n\tau+m))}{\vartheta_1(\tau,r_L(n\tau+m))\vartheta_1(\tau,r_L(n\tau+m))}\ ,
\end{align} where $M=\gcd(n,m)$. When $r_LM= 0\mod \ZZ$,  ${\cal Z}_{g^n,g^m}$ is given by a suitable $\SL_2(\ZZ)$ transformation of \eqref{alternative}.

The elliptic genus of the $g$-orbifolded theory $\theory_{K3}$, which we assume to be a $K3$ model, is then given in the usual way by
\be {\cal Z}(\theory_{K3};\tau,z)=\frac{1}{N}\sum_{n,m\in\ZZ/N\ZZ} {\cal Z}_{g^n,g^m}(\theory_{T^4};\tau,z)\ .
\ee 
Similarly, the twining genus of the quantum symmetry $Q$ is given by 
\be\label{twinQ} {\cal Z}_Q(\theory_{K3};\tau,z)=\frac{1}{N}\sum_{n,m\in\ZZ/N\ZZ}  e^{\frac{2\pi i n}{N}} {\cal Z}_{g^n,g^m}(\theory_{T^4};\tau,z)\ .
\ee  
A number of new twining genera can be obtained from the above calculation. The relevant values of $r_L,r_R$ and the Frame shapes of the corresponding quantum symmetries are collected in Table \ref{t:Qtwin}.

A set of more general twining genera can be obtained as follows. 
 Suppose that $g$ is a symmetry of a NLSM on $T^4$ of order $N$ and the $g^n$-orbifolded theory is a $K3$ NLSM for a $n|N$.
Then $g$ induces a symmetry $g'$ of order $N/n$ on the resulting $K3$ NLSM that commutes with the quantum symmetry, and one has
\be\label{gT4orbif} {\cal Z}_{g'^lQ^m}(\theory_{K3};\tau,z)=\frac{n}{N}\sum_{j,k=1}^{N/n} e^{\frac{2\pi i jm}{N/n}} {\cal Z}_{g^{nj},g^{nk+l}}(\theory_{T^4};\tau,z)\ .
\ee The right-hand side of this equation can be easily computed using \eqref{T4twin}. The Frame shapes corresponding to these symmetries are collected in Table \ref{t:Qtwin2}.

\begin{table}[h]
\centering
\begin{tabular}{|c|c|c|c|}
\hline
 $r_L$ & $r_R$ & $\pi_Q$ & w-s parity\\
\hline
 $1/2$ & $1/2$ & $1^{-8}2^{16}$ & $\circ$\\
\hline
 $1/3$ & $1/3$ & $1^{-3}3^{9}$& $\circ$\\
\hline
 $1/4$ & $1/4$ & $1^{-4}2^{6}4^4$& $\circ$\\
\hline
 $1/6$ & $1/6$& $1^{-4}2^{5}3^46^1$& $\circ$\\
\hline
 $1/5$ & $2/5$& \multirow{2}{*}{$1^{-1}5^{5}$}& \multirow{2}{*}{$\updownarrow$}\\
 $2/5$ & $1/5$&  & \\
\hline
 $1/4$ & $1/2$ &  \multirow{2}{*}{$2^{-4}4^{8}$}& \multirow{2}{*}{$\updownarrow$}\\
 $1/2$ & $1/4$ & &\\
\hline
 $1/6$ & $1/2$ &  \multirow{2}{*}{ $1^{-2}2^{4}3^{-2}6^4$}& \multirow{2}{*}{$\updownarrow$}\\
 $1/2$ & $1/6$ & &\\
\hline
 $1/6$ & $1/3$ & \multirow{2}{*}{ $1^{-1}2^{-1}3^36^3$}& \multirow{2}{*}{$\updownarrow$}\\
 $1/3$ & $1/6$ & &\\
\hline
$1/8$ & $5/8$ & \multirow{2}{*}{ $1^{-2}2^{3}4^18^2$}& \multirow{2}{*}{$\updownarrow$}\\
 $5/8$ & $1/8$ & &\\
\hline
 $1/10$ & $3/10$ & \multirow{2}{*}{ $1^{-2}2^{3}5^210^1$}& \multirow{2}{*}{$\updownarrow$}\\
 $3/10$ & $1/10$ & &\\
\hline
 $1/12$ & $5/12$ & \multirow{2}{*}{ $1^{-2}2^{2}3^24^112^1$}& \multirow{2}{*}{$\updownarrow$}\\
 $5/12$ & $1/12$ & &\\
\hline
\end{tabular}
\caption{\footnotesize{Frame shapes corresponding to quantum symmetries $Q$ of torus orbifolds. The twining genera can be obtained by applying formulae \eqref{twinQ} and \eqref{T4twin}.
Clearly, a twining genus for a quantum symmetry is fixed by world-sheet parity if and only if $r_L=\pm r_R\mod \ZZ$.
}}\label{t:Qtwin}
\end{table}

\begin{table}[h]
\centering
\begin{tabular}{|c|c|c|c|c|}
\hline
 $r_L$ & $r_R$ & $o(Q)$ & $o(g')$ & $\pi_{Qg'}$\\
\hline
$1/4$ & $1/4$ & $2$& $2$ & $2^{12}$\\
\hline 
\multirow{2}{*}{$1/6$} & \multirow{2}{*}{$1/6$} & $2$& $3$ & $1^4 2^1 3^{-4} 6^5$\\
 & & $3$ & $2$& $ 1^5 2^{-4} 3^1 6^4$\\
\hline 
$1/6$ & $1/2$ &  $ 2$& $3$ & $1^{-2} 2^4 3^{-2} 6^4$\\
\hline 
$1/6$ & $1/3$ &  $3$ & $2$ & $ 1^{-1} 2^{-1} 3^3 6^3$ \\
\hline 
\multirow{2}{*}{$1/8$} & \multirow{2}{*}{$3/8$} & 2& 4 & $2^44^4$\\
& & 4& 2 & $2^44^4$\\
\hline
\multirow{2}{*}{$1/10$} & \multirow{2}{*}{$3/10$} & $2$ & $5$ & $1^22^15^{-2}10^3 $\\
   &  & $5$ & $2$ & $1^3 2^{-2} 5^1 10^2 $\\
\hline
\multirow{4}{*}{$1/12$} & \multirow{4}{*}{$5/12$} & $6$ & $2$ & $2^36^3$\\ 
&& $4$ & $3$ & $1^23^{-2}4^{1}6^212^1$\\
&& $3$ & $4$ & $1^12^23^14^{-2}12^2$\\
&& $2$ & $6$ & $2^36^3$\\
\hline
\end{tabular}
\caption{\footnotesize{Symmetries of torus orbifolds whose twining genera are given by  \eqref{gT4orbif}. }}\label{t:Qtwin2}
\end{table}

\subsection{Landau-Ginzburg orbifolds}\label{LGtwin}

It is very non-generic for a $K3$ NLSM to correspond to an exactly solvable CFT. 
In fact, the only such examples we know of are torus orbifolds, described in the previous subsection, Gepner models, i.e. orbifolds of tensor products of $\mathcal N=2$ minimal models \cite{Gepner:1987qi}, and generalizations thereof \cite{Font:1989qc}.
However, for the purpose of computing the (twined) elliptic genus, it is sufficient to have a UV description which flows in the IR to a $K3$ NLSM. This fact was used by Witten to provide evidence for the connection between certain Landau-Ginzburg (LG) models and $\mathcal N=2$ minimal models \cite{Witten:1993jg}. The LG theories are generically massive, super-renormalizable $\mathcal N=2$ quantum field theories; however, in the IR they can flow to an $\mathcal N=(2,2)$ superconformal field theory. For instance, the LG theory of a single chiral superfield with superpotential
\be
W_{A_{k+1}}(\Phi)= {1\over k+2}\Phi^{k+2}
\ee
flows to an IR fixed point corresponding to the $\mathcal N=2$ minimal model of type $A_{k+1}$.

Though these minimal models all have central charge less than 3, LG theories prove to have geometric applications through the orbifold construction. Namely, 
one can construct theories which flow in the IR to a NLSM on a CY $d$-fold by taking superpotentials of multiple chiral multiplets, such that the sum of their charges equals $3d$,  along with an orbifold which projects the Hilbert space onto states with integer $U(1)$ charges. 
This connection between CY geometry and LG orbifolds was further elucidated by Witten \cite{phases} using the framework of gauged linear sigma models.

In \cite{cheng2015landau}, a number of new twinings were found in explicit models: LG orbifolds which flow in the IR to $K3$ CFTs. 
Here we briefly mention cases where symmetries of order 11, 14, and 15 arise. These symmetries preserve precisely a four-plane in the Leech lattice, and thus only occur at isolated, nonsingular points in $K3$ moduli space. The symmetries of order 11 and 15 arise in cubic superpotentials of six chiral superfields of the form,
\begin{align}
  \mathcal{W}^c_1(\Phi)&=\Phi_0^3 + \Phi_1^2 \Phi_5 + \Phi_2^2 \Phi_4 + \Phi_3^2 \Phi_2 + \Phi_4^2 \Phi_1 + \Phi_5^2 \Phi_3 \\
   \mathcal{W}^c_2(\Phi)&= \Phi_0^2 \Phi_1 + \Phi_1^2 \Phi_2 + \Phi_2^2 \Phi_3 + \Phi_3^2 \Phi_0 + \Phi_4^3 + \Phi_5^3,
\end{align}
respectively, while the symmetry of order 14 arises in a model with  quartic superpotential 
\be
\mathcal{W}^q(\Phi)=\Phi_1^3 \Phi_2 + \Phi_2^3 \Phi_3 + \Phi_3^3 \Phi_1 + \Phi_4^4. 
\ee
As discussed in \cite{cheng2015landau}, the symmetry groups of $\mathcal{W}^c_1$, $\mathcal{W}^c_2$, and $\mathcal{W}^q$ are given by 
 $L_2(11)$, $(3 \times A_5):2$ and $L_2(7)\times 2$, each of which contains elements of order 11, 15, and 14, respectively. 
Using their explicit actions on the superfields one can readily compute their LG twining genus. 

The symmetries of order 11, 15, and 14 all have a unique Frame shape ($1^211^2$, $1.3.5.15$ and $1.2.7.14$ respectively) and each occur in two non-Conway Niemeier groups,  corresponding to Niemeier lattices $N_1, N_2$ with root lattices $\{A_1^{24}, A_2^{12}\}$, $\{A_1^{24}, D_4^6\}$ and  $\{A_1^{24}, A_3^8\}$ respectively.
Since these symmetries preserve exactly a four-plane, the Conway module associates two different twinings functions to these Frame shapes. 
In each of these three cases, the two umbral moonshine twinings given corresponding to two Niemeier lattices yield two different results $\phi^{N_1}_{g_1}$ and $\phi^{N_2}_{g_2}$, coinciding with the two twinings $\phi^\Leech_{g,+}$ and $\phi^\Leech_{g,-}$ arising from Conway module.

The twinings of order 11, 15 and 14 computed in the above-mentioned LG models match those associated with root systems $A_2^{12}$, $D_4^6$ and $A_3^8$, respectively. 
This can be viewed as evidence for the connection between (non-$M_{24}$ instances of) umbral moonshine, as well as Conway moonshine, to the symmetries of $K3$ NLSMs 
\footnote{It is intriguing to note that the forms of $ \mathcal{W}^c_1$, $ \mathcal{W}^c_2$ and ${\mathcal W}^q$ are closely related to the superpotentials which flow to the $A_2$, $D_4$ and $A_3$ $\mathcal N=2$ minimal models, where the $A$-type case is given above, and the $D_4$ case is
$
W_{D_4}(\Phi_1, \Phi_2)\sim\Phi_1^3+\Phi_1\Phi_2^2.
$ It would be interesting to understand if this is connected to the fact that the  twinings correspond to cases of umbral moonshine whose root systems contain copies of $A_2$, $D_4$ and $A_3$, respectively.
}.  We refer to \cite{cheng2015landau} for more examples and details.

\subsection{Modularity}\label{s:twin_from_modul}

In this section we discuss how one can use constraints of modularity to precisely specify the twining genera corresponding to certain $O^+(\Gamma^{4,20})$ conjugacy classes in some cases. The twining genera $\eg_g$ are weak Jacobi forms under some congruence subgroup $\modgrp \subseteq \SL_2(\ZZ)$, possibly with a non-trivial multiplier $\psi$. 
At the same time, the Frame shape establishes the $q^0$-terms in their Fourier expansion, given by
\be\label{normalize1} {\cal Z}_g(\tau,z)=2y+2y^{-1}+\Tr_{V_{24}}(g)-4+O(q)\ . 
\ee
Here, $\Tr_{V_{24}}(g)$ denotes  the trace of $g\in O^+(\Gamma^{4,20})$ in the defining $24$-dimensional representation $V_{24}$. 
In some cases, the modular properties together with the above leading term coefficients  
are sufficient to fix the function $\eg_g$ completely. 
More precisely, the above criteria dictate that $\eg_g$ 
can be written as
\be \eg_g(\tau,z)=\frac{\Tr_{V_{24}}(g)}{12}\phi_{0,1}(\tau,z)+F(\tau)\phi_{-2,1}(\tau,z)\ ,
\ee where \begin{align} 
\phi_{0,1}(\tau,z)&=4\sum_{i=2}^4 \frac{\theta_i(\tau,z)^2}{\theta_i(\tau,0)^2}=y+10+y^{-1}+O(q)\\ \phi_{-2,1}(\tau,z)&=\frac{\theta_1(\tau,z)^2}{\eta(\tau)^6}=y-2+y^{-1}+O(q)\ ,\end{align} are the standard weak Jacobi forms of index $1$ and weight $0$ and $-2$, respectively,  and \be\label{normalize2} F(\tau)=2-\frac{\Tr_{\bf{24}}(g)}{12} +O(q)\ ,\ee is a modular form of weight $2$ under $\modgrp$, with a suitable multiplier $\psi$. 
Clearly, $\psi$ can only be non-trivial when $\Tr_{\bf{24}}(g)=0$. 
Let us denote by $M_2( \modgrp;\psi)$ the space of modular forms of weight $2$ for a group $\modgrpnog\in \SL_2(\ZZ)$ with multiplier $\psi$. It is clear from \eqref{normalize1} and \eqref{normalize2} that $\eg_g$ is uniquely determined in terms of $\Tr_{\bf 24}(g)$ whenever $\dim M_2(\modgrp;\psi)\le 1$.

The approach described above is particularly effective in constraining twining genera with non-trivial multiplier $\psi$, since the space $M_2(\modgrp;\psi)$ is often quite small.  We illustrate our arguments with the following example. 
 Consider $g$ with Frame shape $3^8$. The possible multipliers can be determined using the methods described in appendix \ref{a:modularity}. In particular, $\Tr_{V_{ 24}}(g)=0$ and $\modgrp=\Gamma_0(3)$, and hence the order of the multiplier system is either $1$ or $3$. The Witten index of a putative orbifold by $g$ is $8$, which is different from $0$ or $24$. 
We can therefore conclude that the orbifold is inconsistent and hence the multiplier has order $n=3$. (See appendix \ref{a:modularity} for the detailed argument.)
 Thus,  $F(\tau)=2+O(q)$ is modular form of weight $2$ for $\Gamma_0(3)$ with multiplier of order $3$. It turns out that there are two possible multipliers $\psi$ and $\bar \psi$ of order $3$, with the property $\dim M_2(\Gamma_0(3);\psi)=\dim M_2(\Gamma_0(3);\bar\psi)=1$. Hence, in both cases there is a unique weight 2 form $F$, and therefore a unique weak Jacobi form $\eg_g$, with the required normalization   \eqref{normalize2}, giving the umbral twining function corresponding to the root systems ${A_1^{24}}$ and ${A_2^{12}}$.  

Using similar arguments, one can determine the twining genera for the Frame shape $4^6$ for both possible choices of multipliers, and the twining genera for the Frame shapes $6^4$ and $4^28^2$ for one of the two possible multipliers. In all such cases, the resulting twining genera coincide with some umbral functions, i.e. some $\Phi(N)$ (see appendix \ref{a:results}), offering support for our Conjecture \ref{conj:all_arise_from_moonshine}.

\subsection{Evidence for the Conjectures}\label{s:evidence}
In this section we summarize a number of  results which we view as compelling evidence for our conjectures of \S\ref{s:conjectures}. 
Conjecture \ref{conj:all_arise_from_moonshine} states, among other things, that all physical twining genera $\eg_g$ are reproduced by some function arising from umbral or/and Conway module. If true, then combined with the world-sheet parity analysis in \S\ref{s:parity},  the following two statements necessarily  hold. The fact that they do hold then constitutes non-trivial evidence for the conjecture. 

\begin{itemize}
\item As reported in appendix D.2, there are either 81 or 82 distinct $O^+(\Gamma^{4,20})$ classes of symmetries. In particular, for the Frame shape $1^{-4} 2^5 3^4 6^1$, there is either a single $O^+(\Gamma^{4,20})$ class or two classes that are the inverse of each other and hence must have the same twining genus. Therefore, there are potentially $81$ distinct twining genera $\eg_g$. Only $56$ have been computed using the methods described in \S\ref{geomtwin}-\ref{s:twin_from_modul}. In all such cases, one has $\eg_g\in \Phi(N)$
 for at least one Niemeier lattice $N$.
 
\item Whenever there is an umbral or Conway twining genus $\phi_g\in \Phi(N)$ which has a complex multiplier $\psi$, there exists another $\phi'_{g'} \in \Phi(N')$ corresponding to the same Frame shape $\pi_g=\pi_{g'}$ and with the conjugate multiplier $\bar\psi$. 
Furthermore, $\pi_g$ has distinct $O^+(\Gamma^{4,20}) $ conjugacy classes which are related by world-sheet parity. Note that in all cases we have $N\neq N'$. 
Table \ref{tbl:mult} shows the pairs of $N,N'$, denoted in terms of their root systems in the case $N\neq \Leech$, leading to Jacobi forms with complex conjugate multipliers.
\begin{center}
\begin{table}[htb]
\begin{center}
\begin{tabular}{c|c|c}
Frame shape & $\psi$ & $\overline{\psi}$\\\hline
$3^8$& $A_2^{12}, D_4^6, A_8^3, E_8^3, \Lambda$& $A_1^{24}, A_4^6, D_8^3$ \\\hline
$4^6$ & $A_3^8, A_4^6, A_{12}^2, \Lambda$& $A_1^{24}, A_2^{12}, A_6^4, D_6^4$ \\\hline
$6^4$ & $A_2^{12}, D_4^6, A_8^3, \Lambda$& $A_1^{24}, A_4^6$ \\\hline
$4^28^2$ & $A_3^8, E_6^4, \Lambda$& $A_2^{12}$ 
\end{tabular}
\end{center}\caption{\footnotesize{Frame shapes with complex multiplier and corresponding Niemeier lattices. 
}}\label{tbl:mult}
\end{table}
\end{center}
\end{itemize}
Similarly, the following fact is non-trivially compatible with Conjecture \ref{c:manytwin}. 
\begin{itemize}
\item Fix a four-plane preserving Frame shape $\pi_g$.  
Denote by $K$ the number of distinct twining functions $\phi_g^N$ associated with $\pi_g$ arising from either Conway or umbral moonshine, and denote by $K'$ the number of $O^+(\Gamma^{4,20}) $ conjugacy classes associated with $\pi_g$. In all cases, $K' \geq K$, and for a vast majority  ($35$ out of $42$) of the four-plane preserving Frame shapes this inequality is saturated. 
\end{itemize}

Note that the fact that $K$ is small can be attributed to the large number of coincidences among the elements of $\Phi(N)$ and $\Phi(N')$ related to different Niemeier lattices $N, N'$. For example, the Frame shape $2^44^4$ appears in the group $G_N$ for seven distinct Niemeier lattices $N$, but the seven twining genera $\phi_g^N$ are all the same, compatible with the fact that there is a unique $O^+(\Gamma^{4,20})$ class for this Frame shape.
Since for some  ($7$ out of $42$)  Frame shapes the number of $O^+(\Gamma^{4,20})$-classes is strictly larger than the number of distinct $\phi_g^N$, Conjecture \ref{conj:all_arise_from_moonshine} predicts that there must be some coincidences among the physical twining genera corresponding to these different classes.

\section{Discussion}\label{sec:discussion}

In the present paper we have proven classification results on lattices and groups relevant for symmetries of $K3$ string theory and proposed conjectures regarding the relation between the these symmetries and umbral and Conway moonshine. 
These results motivate a number of interesting questions. 
We discuss a few of them here. 

\begin{itemize}[leftmargin=*]
\item{
Apart from classifying the symmetry groups of $K3$ NLSMs as abstract groups, it is also important to know what their actions are on the (BPS) spectrum. In particular, the twining genus can differ for two $K3$ NLSM symmetries with the same embedding into $\Co_0$ \cite{K3symm,Cheng:2014zpa,cheng2015landau}. 
 This motivated us to classify the 
distinct conjugacy classes in $O^+(\Gamma^{4,20})$ and $O(\Gamma^{4,20})$ for a given four-plane preserving Frame shape. 

Given this consideration and given our Conjectures \ref{conj:all_arise_from_moonshine} and \ref{c:manytwin} relating twining genera and moonshine functions, an important natural question is the following: given a particular $K3$ NLSM, 
 how do we understand which case(s) of umbral moonshine govern its symmetries?
} 
\item{
In this paper we  extend the classification of symmetry groups to singular points in the moduli space of $K3$ NLSMs. These singular points correspond to perfectly well-defined string compactifications where the physics in the six-dimensional non-compact spacetime involves enhanced non-abelian gauge symmetries. It will be interesting to study the BPS-counting functions arising in these compactifications.

Moreover, as these points are T-dual to type IIB compactifications on K3 in the presence of an NS5-brane  \cite{Witten:1995zh}, it would be interesting to explore the symmetries of these special points from this spacetime point of view. Furthermore, it may also be interesting to classify the symmetry groups in more general fivebrane spacetimes, such as those studied in \cite{Harvey:2013mda,Harvey:2014cva}  in connection with umbral moonshine. 
}

\item{More generally, one can try to classify the discrete symmetry groups which arise in other supersymmetric string compactifications, in varying dimensions and with differing numbers of supersymmetries. For example, one case of particular interest is the symmetries of theories preserving only eight supercharges. One difficulty in studying such theories is the the global form of the moduli space is often not known, so one does not have the power of lattice embedding theorems used to study theories with sixteen supercharges. However, it may be possible to get partial results in certain examples. The connection between sporadic groups, geometry, and automorphic forms in theories with eight supercharges has only somewhat been studied (see, for e.g., \cite{Cheng:2013kpa,Harrison:2013bya}) and it would be interesting to  explore it further. }

\item{Twining genera of $K3$ NLSMs can be lifted to twining genera of the $N$th symmetric product CFT $Sym^N(K3)$ through a generalization \cite{Cheng:2010pq} of the formula for the symmetric product elliptic genus of \cite{Dijkgraaf:1996xw}. It can happen that a symmetry which is not a geometric symmetry of any $K3$ surface can be a geometric symmetry for a hyper-K\"ahler manifold that is deformation equivalent to the  $N$-th Hilbert scheme of a $K3$ surface for $N\geq 2$. The symmetries of such hyper-K\"ahler manifolds of $K3^{[N]}$ type were classified in \cite{Hoehn:2014ika} for $N=2$ in terms of their embedding into $Co_0$. This includes Frame shapes corresponding to elements of order 3, 6, 9, 11, 12, 14, and 15 which are not geometric symmetries of any $K3$ surface.  Each of these elements has at least two distinct twining functions associated with it via umbral and Conway moonshine as presented in Table \ref{last_table_genera}. We noticed that for the elements of order 11, 14, and 15 that these distinct twining functions lift to the same twined elliptic genus for $Sym^N(K3)$ for $N=2,3,4.$ It would be interesting to understand when this general phenomenon occurs, and more generally the structure of symmetries of string theory on $K3 \times S^1$.}

\item The compactification of type IIA on $K3\times T^2$ gives rise to a four dimensional model with half-maximal supersymmetry ($16$ supercharges). When the internal NLSM has a symmetry $g$, one can construct a new four dimensional model (CHL model) with the same number of supersymmetries  \cite{Chaudhuri:1995fk,Chaudhuri:1995dj,Sen:1995ff,Chaudhuri:1995bf}. The CHL model is defined as the  orbifold of type IIA on K3$\times T^2$ by a fixed-point-free symmetry acting as $g$ on the K3 sigma model and, simultaneously,  by a shift along a circle $S^1$ in the $T^2$. The twining genus $\eg_g$ is directly related to the generating function $1/\Phi_g$ of the degeneracies of 1/4 BPS dyons in this CHL model \cite{Dijkgraaf:1996it,Dijkgraaf:1996xw,Shih:2005uc,Jatkar:2005bh,David:2006ji,David:2006ru,David:2006ud,David:2006yn,Cheng:2008kt}. Up to dualities, the CHL model only depends on the Frame shape of $g$ \cite{Persson:2015jka}. This is apparently puzzling for those Frame shapes that correspond to multiple $O^+(\Gamma^{4,20})^+$-classes and therefore to multiple twining genera $\eg_g$: in these cases, there are different candidates $1/\Phi_g$ for the 1/4 BPS counting function, one for each distinct twining genus $\eg_g$. Since $O^+(\Gamma^{4,20})$ is part of the T-duality group of the four dimensional model, a natural interpretation of this phenomenon is that the different $1/\Phi_g$ functions count 1/4 BPS dyons related to different T-duality orbits of charges in the same CHL model. In view of this interpretation, it would be interesting to understand the precise correspondence between $O^+(\Gamma^{4,20})^+$-classes and T-duality orbits of charges.

\item 
One piece of supporting evidence for our conjectures concerns twining genera with complex multiplier systems. 
However, so far we have not been able to directly obtain these proposed twining genera from $K3$ NLSMs. Nevertheless, we argue that this is unsurprising and does not constitute discouraging counter evidence for our conjectures for the following reason. Recall that the argument in \S \ref{s:parity} indicates that these functions must arise from a symmetry acting differently on left- and right-movers. Then our Conjecture \ref{conj:all_arise_from_moonshine}, together with the observation that such twining functions always arise from multiple instances of umbral and Conway moonshine (see \S\ref{s:evidence}, 2nd bullet point), predicts that these theories correspond to lattices  embeddable into multiple Niemeier lattices. This precludes most of the exactly solvable models that have been studied so far, in particular all torus orbifolds and some Gepner models, since these always contain a quantum symmetry which can only arise from a Leech embedding. So far most of the NLSM analysis has focussed on these exactly solvable models, and this explains why we have not observed these proposed twining genera yet. 

On the other hand, a number of the proposed twining genera with complex multipliers (as well as many with real multipliers)  were found by twining certain LG orbifold theories \cite{cheng2015landau}.
These include functions arising from symmetries of order $3,4,6$ and $8$ and with Frame shapes $3^8, 4^6,6^4$ and $4^28^2$--the four Frame shapes which both preserve a four-plane in $Co_0$ and correspond to twining genera with complex multiplier. 
In order to obtain these twining genera, one has to consider symmetries which act asymmetrically on the left- and right-moving fermions in the chiral multiplets, such that the UV Lagrangian, the right-moving $\mathcal N=2$ algebra, and the four charged Ramond ground states are preserved. In general, however, the left-moving $\mathcal N=2$ algebra is not preserved. Though $H_L$ and $J_L$ must remain invariant for the twining genus to be well defined, $G_-$ and $\overline G_-$ are transformed under these symmetries, such that the symmetry maps the left-moving $\mathcal N=2$ to a different but  isomorphic copy. See \cite{cheng2015landau} for more details. 
It is important to note that, though these symmetries do not preserve the full UV supersymmetry algebra, it does not preclude the possibility that they preserve a copy of the IR $\mathcal N=(4,4)$ SCA. After all, there is only an $\mathcal N=(2,2)$ supersymmetry algebra apparent in the UV, and only after a non-trivial RG flow involving a complicated renormalization of the K\"ahler potential does the symmetry get enhanced to $\mathcal N=(4,4)$ at the conformal point. A clarification of the IR aspects of these UV symmetries would be helpful in unravelling the nature of these left-right asymmetric symmetries. 

\item{
While our Conjecture \ref{c:manytwin} states that all umbral and Conway moonshine functions corresponding to four-plane preserving group elements play a role in the twining genera of $K3$ NLSMs, the physical relevance of the umbral (including Mathieu) moonshine functions corresponding to group elements preserving only a two-plane remains unclear. We highlight a number of approaches to this problem here.

One possible approach to the problem is to find is to find a way to combine symmetries realized at different points in moduli space and in this way generate a larger group which also contains two-plane preserving elements. 
This approach is motivated by the fact that 
the elliptic genus receives only contributions from BPS states and is invariant across the moduli space. 
This possibility was first raised as a question 
``Is it possible that these automorphism groups at isolated points in the moduli
space of $K3$ surface are enhanced to $M_{24}$ over the whole of moduli space when we consider the elliptic genus?''
in \cite{EOT}. Concrete steps towards realising this idea in the context of Kummer surfaces were taken in \cite{Taormina:2011rr,Taormina:2013mda,cheng2015landau}. See also \cite{Gaberdiel:2016iyz} for recent progress in the direction. 

%One possible approach to this problem, as originally envisioned by the authors of \cite{EOT} and further investigated in \cite{Taormina:2011rr,Taormina:2013mda,cheng2015landau,Gaberdiel:2016iyz}, is 
%to find a way to combine symmetries realized at different points in moduli space and in this way generate a larger group which also contains two-plane preserving elements. 

A second approach is to consider string compactifications  
where larger groups are realized at  given points in moduli space as symmetry groups of the full theory (and not just the BPS sector). 
For theories with 16 supercharges, this is only possible for compactifications with less than six non-compact dimensions.  
For example, it was shown that there are points in 
the moduli space of string theory compactifications to three dimensions which admit the Niemeier groups as discrete symmetry groups \cite{Kachru:2016ttg}. 
 In the type IIA frame, these are given by compactifications on $K3\times T^3$. 
The action of these symmetry groups on the 1/2-BPS states of the theory has been analyzed \cite{Kachru:2016ttg}, and  
 it would be interesting to understand the action on the 1/4-BPS states.

A third approach stems from the vertex operator algebra (VOA) perspective. In \cite{M5}, a close variant of the Conway module is shown to exhibit an action of a variety of two-plane preserving subgroups of $\Co_0$, including $M_{23}$, and yields as twining genera a set of weak Jacobi forms of weight zero and index two.\footnote{Coincidentally, in \cite{Benjamin:2014kna}, and as further discussed in \cite{Cheng:2015fha}, it was shown that this module also admits an $M_{24}$ action, although the twining genera are no longer weak Jacobi forms the representations are less closely related to those of Mathieu moonshine.} In addition, the mock modular forms which display $M_{23}$ representations appear to be very closely related to the mock modular forms which play a role in $M_{24}$ moonshine. However, the physical relevance of this module is still unclear. A better understanding of the connection between the Conway module and $K3$ NLSMs could help explain Mathieu and umbral moonshine. 

Finally, yet another  approach is to consider compactifications preserving less supersymmetry \cite{Cheng:2013kpa,Harrison:2013bya}. It is not unlikely that the ultimate explanation of umbral moonshine will require a combination of the above approaches. }

\end{itemize}

\section*{Acknowledgements}

We thank John Duncan, Francesca Ferrari, Matthias Gaberdiel, Gerald H\"ohn, Shamit Kachru, and Natalie Paquette for discussions about related subjects. 
The work of M.C. is supported by ERC starting grant H2020 ERC StG 2014.  S.M.H. is supported by a Harvard University Golub Fellowship in the Physical Sciences and DOE grant DE-SC0007870. RV is supported by a grant from `Programma per giovani ricercatori Rita Levi Montalcini'. M.Z. is supported by the Mellam Family Fellowship at the Stanford Institute for Theoretical Physics.

\appendix
\section{Some Basic Facts about Lattices }\label{app:lats}
In this appendix we collect some useful facts about lattices (see \cite{Nikulin,ConwaySloane} for more details). Let $L$ be an even lattice with non-degenerate quadratic form $Q:L\to 2\ZZ$ of signature $(n,m)$ and rank $d=n+m$. 
Let $L^*$ be the dual, so that naturally $L\subseteq L^*$, $L^*\subset L\otimes_\ZZ \QQ$. The quadratic  form $Q$ extends to $Q:L^*\to \QQ$ by linearity. The discriminant group of $L$
\be A_L:=L^*/L\ ,
\ee is a finite abelian group, with a discriminant form
\be q_L:A_L\to \QQ/2\ZZ
\ee induced by $Q:L^*\to \QQ$. The discriminant group has order $\det B_{ij}$ where $B_{ij}$ is the Gram matrix for the quadratic form $Q$ with respect to some basis $v_1,\ldots,v_d$ of $L$. 

We denote by $O(L)$ the group of automorphisms of $L$ (which is the same as the group $O(L^*)$ of automorphisms of the dual $L^*$). Similarly, we define $O(q_L)$ as the group of automorphisms of the discriminant group $A_L$ that preserves the quadratic form $q_L$
\be O(q_L):=\{g\in \Aut(A_L)\mid q_L(g(x))=q_L(x)\text{ for all }x\in A_L\}\ .
\ee There is an obvious map
\be O(L)\to O(q_L)\ ,
\ee which, in general, is neither injective nor surjective.

The information including the signature $(n,m)$ of the lattice $L$ and the isomorphism class of its discriminant form $q_L$ amounts to giving the genus of $L$. More formally: two lattices $L,L'$ are in the same genus if and only if they have the same signature and the discriminant forms are isomorphic, i.e. there is an isomorphism $\gamma:A_L\stackrel{\cong}{\to} A_{L'}$ such that $q_{L'}\circ \gamma=q_L$ \cite{Nikulin}.

Let $L$ be a unimodular lattice and $M$ a primitive sublattice. Recall that $M\subset L$ is said to be a primitive sublattice if the following three equivalent statements are true 1. $L/M$ has no torsion; 2. $(M\otimes \QQ)\cap L=M$; 3. if $v\in L$ and $nv\in M$ for some $n\in\ZZ$, then $v\in M$.    
 Let $N$ be the orthogonal complement of $M$ (then, $N$ is automatically primitive), so that
\be N\oplus M\subseteq L\subseteq N^*\oplus M^*\ .
\ee
 Then, there is an isomorphism $\gamma:A_M\to A_N$ of discriminant groups, such that
\be\label{opposite} q_N\circ \gamma = -q_M\ ,
\ee and satisfying
\be\label{gluing} L=\{(v,w)\in M^*\oplus N^*\mid \gamma(\bar v)=\bar w \}\ ,
\ee where $\bar v\in A_M$ (resp., $\bar w\in A_N$) is the class with representative $v\in M^*$ (resp., $w\in N^*$). Vice versa, given two even lattices $N,M$ with an isomorphism $\gamma:A_M\to A_N$ satisfying \eqref{opposite}, then the lattice $L$ defined by \eqref{gluing} 
is an even unimodular lattice, such that $N,M$ are two mutually orthogonal primitive sublattices of $L$.

\section{Proofs of results in \S\ref{sec:class}} \label{app:pfs}

\subsection{Proof of Theorem \ref{t:dirclass}}\label{a:proofdirclass}

Recall that, given an even lattice $T$, the discriminant group is the finite abelian group $A_T=T^*/T$. The quadratic form on $T$ induces a quadratic form $q_T:A_T\to \QQ/2\ZZ$ on the discriminant group, called the discriminant (quadratic) form.

Theorem 1.12.4 of \cite{Nikulin} gives sufficient conditions for the existence of a primitive embedding of an even lattice $T$ of signature $(t^+,t^-)$ into some even unimodular lattice $L$ of signature $(l^+,l^-)$, with $l^+-l^-\equiv 0\mod 8$:
\be\label{suffglue} \begin{cases}l^+\ge t^+,\ l^-\ge t^-\\ t^++t^-\le \frac{1}{2}(l^++l^-)\end{cases}\qquad \Rightarrow\qquad \text{exists primitive } T\hookrightarrow L\ .\ee Alternatively, a necessary and sufficient condition for such an embedding is the existence of a lattice $K$ of signature $(l^+-t^+,l^--t^-)$ such that 
\be\label{gluingcond} q_K\cong -q_T
\ee
where $q_K$ and $q_T$ are the discriminant forms of $K$ and $T$ (see \cite{Nikulin}, Theorem 1.12.2). More precisely, when \eqref{gluingcond} is satisfied, one can construct an even unimodular lattice $L$ such that
\be K\oplus T\subseteq L\subseteq K^*\oplus T^*\ ,
\ee and such that the embeddings $T\hookrightarrow L$ and  $K\hookrightarrow L$ are primitive.
Conversely, if $T$ is a primitive sublattice of an even unimodular lattice and $K$ its orthogonal complement, then \eqref{gluingcond} is satisfied.

Since $\Gamma^G$ has signature $(4,d)$, by \eqref{suffglue} it can be primitively embedded into an even unimodular lattice $\Gamma^{8+d,d}$ of signature $(8+d,d)$. Let $S$ be its orthogonal complement in $\Gamma^{8+d,d}$ and $S(-1)$ the lattice obtained by flipping the sign of the quadratic form of $S$. Then, $S(-1)$ has signature $(0,4+d)$ and, using \eqref{gluingcond} repeatedly, we obtain
\be q_{S(-1)}=-q_S\cong q_{\Gamma^G}\cong -q_{\Gamma_G}\ .
\ee Thus, there exists an even unimodular lattice $N$ of signature $(0,24)$ such that
\be \Gamma_G\oplus S(-1)\subseteq N\subseteq \Gamma_G^*\oplus S^*(-1)
\ee and such that the embedding $\Gamma_G\hookrightarrow N$ is primitive. For the proof of the other claims, see appendix B in \cite{K3symm}.

\subsection{Proof of Theorem \ref{t:invclass}}\label{a:proofinvclass}

The proof is completely analogous to the one of theorem \ref{t:dirclass}. Since $N^G$ has signature $(0,4+d)$, by \eqref{suffglue} it can be primitively embedded into an even unimodular lattice $\Gamma^{d,8+d}$ of signature $(d,8+d)$. Let $S$ be its orthogonal complement in $\Gamma^{d,8+d}$ and $S(-1)$ the lattice obtained by flipping the sign of the quadratic form of $S$. Then, $S(-1)$ has signature $(4,d)$ and, using \eqref{gluingcond} repeatedly, we obtain
\be q_{S(-1)}=-q_S\cong q_{N^G}\cong -q_{N_G}\ .
\ee Thus, since there is a unique (up to isomorphism) even unimodular lattice of signature $(4,20)$, we have
\be N_G\oplus S\subseteq \Gamma^{4,20}\subseteq N_G^*\oplus S^*
\ee and the embedding $N_G\hookrightarrow \Gamma^{4,20}$ is primitive. For the proof of the other claims, see appendix B in \cite{K3symm}.

\subsection{Proof of Proposition \ref{t:UmbralSymm}}\label{a:proofUmbralSymm}

Recall that a Weyl chamber $\W\subset N\otimes \RR$ is the closure of any of the connected components in
\be N\otimes \RR \setminus \Bigl(\bigcup_{r\in N,\ r^2=-2} r^\perp\Bigr) \ ,
\ee the complement of the hyperplanes orthogonal to the roots. The Weyl group $W_N$ acts by permutations on the set of Weyl chambers; in particular, for any non-trivial $w\in W_N$ and Weyl chamber $\W$, the interior $\W^o$ and its image $w(\W^o)$ have no intersection
\be W^o\cap w(W^o)=\emptyset\ .
\ee Suppose that $\hat G\cap W_N$ contains some non-trivial element $w$. Since $w$ fixes $N^{\hat G}$ pointwise, it follows that the sublattice $N^{\hat G}$ cannot contain any vector in the interior of a Weyl chamber. Therefore,
\be N^{\hat G}\subset \bigcup_{r\in N,\ r^2=-2} r^\perp\ .
\ee Since $N^{\hat G}\otimes \RR$ is convex, it must be actually contained in some hyperplane $r^\perp$, for some root $r\in N$, which implies that $r\in N_{\hat G}\cong \Gamma_G$. Vice versa, if $N_{\hat G}\cong \Gamma_G$ contains a root $r$, then the corresponding reflection $w_r\in W_N$ fixes $N^{\hat G}$ and therefore is in $G$. For the second point, notice that the natural projection $O(N)\to O(N)/W_N$ defines a homomorphism from $\hat G$ into $G_N=O(N)/W_N$. As we just proved, if $\Gamma_G$ has no roots, then $\hat G$ has trivial intersection with $W_N$, so that this homomorphism is injective. We hence conclude that $\hat G$ is isomorphic to a subgroup of $G_N$.

\section{Modular groups and multipliers}\label{a:modularity}

In this appendix we present the arguments we employed in \S\ref{sec:conj} to determine the modular properties, in particular the multipliers, of the twining genera. 
The twining genera $\eg_g$ are weak Jacobi forms of weight $0$ and index $1$ for some subgroup $\modgrp\subseteq SL_2(\ZZ)$, possibly with a multiplier (group homomorphism)\ $\psi:\modgrp\to \CC^*$.

In this appendix, we describe in some more detail the groups $\modgrp$ and the multipliers $\psi$. Consider an order-$N$ symmetry $g$. 
Recall that the modular transformation $\left(\begin{smallmatrix}
a & b\\ c & d
\end{smallmatrix}\right)\in \SL_2(\ZZ)$ transforms the twining genus $\eg_g$ into a `twisted-twining' genus $\eg_{g^c,g^d}$, the trace of $g^d$ over the $g^c$- twisted sector. As a result, 
the subgroup $\modgrp \subseteq \SL_2(\ZZ)$ corresponding to transformations fixing $\eg_g$ (possibly up to a multiplier) is  given by 
\be\label{modgrp} \modgrp:=\{ \left(\begin{smallmatrix}
a & b\\ c & d
\end{smallmatrix}\right)\in \Gamma_0(N)\mid \exists\, h\in O^+(\Gamma^{4,20}) \text{ s.t. } g^d =hgh^{-1}\text{ or } g^d =hg^{-1}h^{-1}\}\ ,
\ee 
and is in particular always contained in 
$$
 \Gamma_0(N) = \Big\{ \left(\begin{smallmatrix}
a & b\\ c & d
\end{smallmatrix}\right)\in \SL_2(\ZZ) \mid c=0\mod{N} \Big\}.
$$
The group $\modgrp$ for each four-plane preserving Frame shape is given in appendix \ref{a:results}.

The order $n$ of the multiplier $\psi$ of $\eg_g$ can be determined from the Frame shape of $g$ under some physically motivated assumptions about its form, which we will now discuss. In general, for (holomorphic) orbifold CFTs, the multiplier is believed to be completely specified in terms of an element of the third cohomology group $H^3(\ZZ_N,U(1))\cong \ZZ_N$, which determines the modular tensor category of modules over the $g$-invariant subalgebra \cite{Dijkgraaf:1989pz,Dijkgraaf:1989hb,Roche:1990hs}.  We will assume that this is true for $K3$ NLSMs  and their orbifolds. Furthermore, we will assume that the
 triviality of the multiplier $\psi$ is the only condition for the $g$-orbifold to be a consistent CFT. 
 See also recent results on generalized umbral moonshine  \cite{Gaberdiel:2012gf,Cheng:2016nto} 
  for more detail about the relations between multipliers and third cohomology.
 Under these assumptions,  the order $n$ of the multiplier of $\eg_g$ is always a divisor of $N$. 
 Second, for any divisor $d|n$, the element $g^d$ has multiplier of order $n/d$. In particular, $g^n$ is the smallest power of $g$ with trivial multiplier. As a result, if the $g^K$-orbifold of the NLSM is again a consistent CFT, necessarily an $\N=(4,4)$ superconformal field theory at central charge $6$, our assumption then dictates that $(K,N)=n$. 

The above considerations, together with an analysis of the Witten index of the orbifold theory which we  now explain, lead to a derivation of  lower and upper bounds on $n$ for a given $g$.
There is a general formula for the Witten index of an orbifold by a cyclic group $\langle g\rangle$. In terms of the Frame shape $\pi_g=\prod_{\ell|N}\ell^{k_\ell}$ of $g$, this is given by
\be\label{WittenOrbifold}  \sum_{\ell|N} {N\over\ell}\, k_{\ell}\ ,
\ee
 (see for example \cite{Gaberdiel:2012um} or \cite{Persson:2015jka} for a proof). On the other hand, the only possible Witten indices of an $\N=(4,4)$ SCFT with central charge $6$ and integer $U(1)$ charges are $0$ or $24$ \cite{Nahm:1999ps}. Therefore, if for some $g$ the putative Witten index \eqref{WittenOrbifold} of the $g$ orbifold is different from $0$ or $24$, then the orbifold is necessarily inconsistent, and the corresponding multiplier must be non-trivial ($n>1$). An upper bound on $n$ can be obtained by noticing that, by \eqref{normalize2}, the multiplier must be trivial ($n=1$) whenever $\Tr_{V_{24}}(g)\neq 0$. By studying which powers of $g$ have a potentially non-trivial multiplier (i.e. $\Tr_{V_{24}}(g^n)=0$) and which powers give a potentially consistent orbifold (i.e. \eqref{WittenOrbifold} is either $0$ or $24$), one derives lower and upper bounds on the order of the multiplier of $g$,  which are sufficient to determine $n$ in all cases. For almost all the four-plane preserving Frame shapes $\prod_{\ell|N}\ell^{k_\ell}$, the order $n$ is the value of the smallest $\ell$ for which $k_\ell\neq 0$. The only exception is $2^{-4}4^8$, which has trivial multiplier, as follows from the fact that it is the square of an element of Frame shape $1^42^{-2}4^{-2}8^4$, which has non-zero trace and therefore $n=1$.

A case by case analysis for all the four-plane preserving Frame shapes shows that the possible orders of a non-trivial multiplier are $n\in\{2,3,4,6\}$. For each such $n$, the possible forms of the corresponding multiplier $\psi$ are in one-to-one correspondence with the classes of order $n$ in the third cohomology group $H^3(\ZZ_N,U(1))\cong \ZZ_N$. In particular, for $n=2$, there is only one possible multiplier $\psi$, while for each $n=3$, $4$ or $6$  there are two possible multipliers $\psi$ and $\bar \psi$, that are complex conjugate to each other. By the arguments in \S\ref{s:parity}, when $n>2$ both $\psi$ and $\bar \psi$ must appear as the multipliers associated with some physical twining genera.

\section{Some Classification Results 
}\label{a:classifclasses}

\subsection{Classification of  $O^+(\Gamma^{4,20})$ Conjugacy Classes} 
Consider a four-plane-preserving element of $O^+(\Gamma^{4,20})$. Its eigenvalues in the 24-dimensional defining representation can be encoded in a Frame shape. See \S\ref{sec:conj}. 
Given such a Frame shape
, in this appendix we compute the number of compatible $O^+(\Gamma^{4,20})$-conjugacy classes. 
More specifically, we  first discuss a theorem (Theorem \ref{th:numclasses}) which we will employ to determine the number of $O^+(\Gamma^{4,20})$  conjugacy classes with a given Frame shape, as recorded in Table \ref{last_table_genera}.

Let $\Lambda$ be the (negative definite) Leech lattice, $\hat g\in O(\Lambda)\cong \Co_0$ be an automorphism with Frame shape $\pi_g$ and fixing a sublattice $\Lambda^{\hat g}$ of rank $ 4+d$, $d\ge 0$, and let $\Lambda_{\hat g}$ (the co-invariant lattice) be the orthogonal complement of $\Lambda^{\hat g}$ in $\Lambda$.
We denote by $O(\Lambda_{\hat g})$ the group of automorphisms of $\Lambda_{\hat g}$ and by
 \be\label{centrg} C_{O(\Lambda_{\hat g})}(\hat g):=\{\hat h\in O(\Lambda_{\hat g})|\hat g\hat h=\hat h\hat g\}
 \ee the centralizer of $\hat g$ in this group.

Given a primitive embedding
\be i:\Lambda_{\hat g}\hookrightarrow \Gamma^{4,20}\ , 
\ee
denote the image  of $i$ and its orthogonal complement by
\be \Gamma_g:=i(\Lambda_{\hat g})\qquad \Gamma^g:=\Gamma^{4,20}\cap (\Gamma_g)^\perp\ .
\ee
The embedding induces an automorphism $ g\in O^+(\Gamma^{4,20})$ which acts as $\hat g$ on the image $i(\Lambda_{\hat g})$ and trivially on its orthogonal complement $i(\Lambda_{\hat g})^\perp\cap \Gamma^{4,20}$. Namely, we have
\be\label{hatgdef}  g\circ i= i\circ \hat g\qquad  g\rvert_{\Gamma^g}={\rm id}\rvert_{ \Gamma^g}\ .
\ee
Note that given $\hat g$ and $i$, the above fixes $ g$ completely.  
Moreover,  as the notation suggests, $\Gamma^g$ is the sublattice of vectors in $\Gamma^{4,20}$ fixed by $ g$. 

The lattice $\Gamma^g$ has signature $(4,d)$ and, by \eqref{opposite}, its discriminant form $q_{\Gamma^g}$ must be the opposite of $q_{\Lambda_{\hat g}}$:
 \be\label{genusGammag} sign(\Gamma^g)=(4,d),\qquad\qquad q_{\Gamma^g}\cong -q_{\Lambda_{\hat g}}\ .
 \ee As a consequence, the genus of $\Gamma^g$ is determined uniquely in terms of $\Lambda_{\hat g}$, independently of the embedding $i$.    We denote by ${cl}(\hat g)$ the set of isomorphism classes of lattices in this genus.

Conversely, every $ g\in O^+(\Gamma^{4,20})$ fixing a sublattice of signature $(4,d)$ with $d\ge 0$ can be obtained in this way: namely, (\ref{hatgdef}) for some $\hat g\in \Co_0$ and some primitive embedding $i$. As discussed in \S\ref{s:distinctgenera}, two such automorphisms $g_1$ and $g_2$ have the same Frame shape 
if and only if they can be induced by the same $\hat g\in \Co_0$, possibly with different embeddings $i_1$, and $i_2$. In particular, if $g_1$ and $g_2$ are conjugated within $O^+(\Gamma^{4,20})$, then they are necessarily induced by the same $\hat  g$. The converse statement is however not true in general: it can happen that two $g_1$ and $g_2$ are induced from the same $\hat g$ (i.e. they have the same Frame shape), but they are not conjugated in $O^+(\Gamma^{4,20})$. The following theorem will enable us to determine the number of $O^+(\Gamma^{4,20})$  conjugacy classes arising from a given Frame shape, for all four-plane fixing elements of $\Co_0$ with the exception of the Frame shape $1^{-4}2^53^46^1$ which we will discuss at the end of this appendix.

 \begin{theorem}\label{th:numclasses}
 Let $cl^+(\hat  g)$ be the set of equivalence classes of lattices with a sign structure in the the genus determined by \eqref{genusGammag}. Then the number of $O^+(\Gamma^{4,20})$ conjugacy classes with Frame shape $\pi_g$ is given by
 \be \sum_{K\in cl^+(\hat  g)} \left| \overline{C}_{O(\Lambda_{\hat g})}(\hat g)\backslash O(q_{\Lambda_{\hat g}})/{\gamma_K}^*(\overline{O^+}(K))\right|\ ,
 \ee where the sum is over a set of representatives for the isomorphism classes in $cl^+(\hat g)$,  $\overline{O^+}(K)$ (respectively, $\overline{C}_{O(\Lambda_{\hat g})}(g)$) is the image of the natural map $O^+(K)\to O(q_K)$ (resp. $C_{O(\Lambda_{\hat g})}(\hat g)\to O(q_{\Lambda_{\hat g}}))$. Furthermore, for each $K$,  $\gamma_K:q_K\xrightarrow[]{\cong} -q_{\Lambda_{\hat g}}$ is an isomorphism and $\gamma_K^*(O(q_K)):=\gamma_KO(q_K)\gamma_K^{-1}=O(q_{\Lambda_{\hat g}})$   the induced identification of the orthogonal groups. The number of classes does not depend on the choice of these isomorphisms.
 \end{theorem}

To understand the theorem, recall that a (positive) sign structure for a lattice $L$ is a choice of orientation of a maximal positive definite subspace in $L\otimes_\ZZ \RR$.\footnote{Similarly, a negative sign structure is given by a choice of orientation of a maximal negative definite subspace in $L\otimes_\ZZ \RR$. In the following we will only consider positive sign structures and simply refer to them as sign structures.}  We denote by $cl^+(\hat  g)$ the set of classes of lattices with sign structure in the genus \eqref{genusGammag}. In other words, two lattices $L_1$, $L_2$ are equivalent if there is an isomorphism $L_1\to L_2$ that preserves the orientation of maximal positive definite subspaces.

The first step in proving this theorem is to determine when two different embeddings $i_1$ and $i_2$ give rise to $g_1$ and $g_2$ that are conjugated in $O^+(\Gamma^{4,20})$.

\begin{lemma}
Let $i_1$ and $i_2$ be primitive embeddings of $\Lambda_{\hat g}$ into $\Gamma^{4,20}$ and $g_1,g_2\in O^+(\Gamma^{4,20})$ be the automorphisms determined by \eqref{hatgdef}. Then $g_1$ and $g_2$ are conjugated in $O^+(\Gamma^{4,20})$ if and only if there exist $s\in C_{O(\Lambda_{\hat g})}(g)$ and $h\in O^+(\Gamma^{4,20})$, such that
\be h\circ i_1= i_2 \circ s\ ,
\ee and in this case $g_2=hg_1 h^{-1}$.
\end{lemma}
\begin{proof}

Suppose there are $h$ and $s$ such that $h\circ i_1= i_2 \circ s$. Then $h$ induces an isomorphism of the sublattices $i_1(\Lambda_{\hat g})$ and $i_2(\Lambda_{\hat g})$ and the orthogonal complements $h(i_1(\Lambda_{\hat g})^\perp)=i_2(\Lambda_{\hat g})^\perp$. This implies that
\be (hg_1 h^{-1})\rvert_{ i_2(\Lambda_{\hat g})^\perp}={\rm id}\rvert_{ i_2(\Lambda_{\hat g})^\perp}\ .
\ee  Furthermore, the condition $g_1 i_1= i_1 \hat  g$ and the analogue for $g_2, i_2$ (we drop $\circ$ from now on) implies
\be hg_1 i_1=h i_1 \hat g=i_2 s \hat g=i_2 \hat gs=g_2 i_2 s=g_2 h i_1\ .
\ee Using again $h(i_1(\Lambda_{\hat g}))=i_2(\Lambda_{\hat g})$,  it follows that
\be (hg_1 h^{-1})\rvert_{ i_2(\Lambda_{\hat g})}={g_2}{}\rvert_{ i_2(\Lambda_{\hat g})}
 .
\ee As a result, since $hg_1h^{-1}$ coincides with $g_2$ both on $i_2(\Lambda_{\hat g})$ and on its orthogonal complement,  they must be the same.

In the other direction, suppose that $g_2=hg_1 h^{-1}$ for some $h\in O^+(\Gamma^{4,20})$. 
It is easy to see that $g_2$ acts trivially on $h(i_1(\Lambda_{\hat g})^\perp)$ and hence $h(i_1(\Lambda_{\hat g})^\perp) =i_2(\Lambda_{\hat g})^\perp $ and $hi_1(\Lambda_{\hat g})=i_2(\Lambda_{\hat g})$. 
From this and the identity 
\be h i_1 \hat g= h g_1 i_1=g_2 hi_1\ .
\ee 
we get
\be i_2^{-1}h i_1 \hat g= i_2^{-1} g_2 hi_1= \hat g  i_2^{-1}hi_1\ .
\ee This implies that the automorphism $s:=i_2^{-1}h i_1\in O(\Lambda_{\hat g})$ commutes with $\hat g$ and, by definition, $hi_1=i_2s$.
\end{proof}

As a result, the $O^+(\Gamma^{4,20})$ conjugacy classes corresponding to a given four-plane-preserving Frame shape $\pi_g$ of $\Co_0$ 
are in one-to-one correspondence with the classes of primitive embeddings $i:\Lambda_{\hat g}\hookrightarrow  \Gamma^{4,20}$, modulo $O^+(\Gamma^{4,20})\times  C_g(O(\Lambda_{\hat g}))$.
 We are now ready to prove Theorem \ref{th:numclasses}.

 \begin{proof} For each isomorphism class in $cl^+(\hat g)$, we choose once and for all an isomorphism $\gamma_K:q_K\stackrel{\cong}{\rightarrow} -q_{\Lambda_{\hat g}}$ and identify $O(q_{\Lambda_{\hat g}})$ with $\gamma_KO(q_K)\gamma_K^{-1}$. Note that any other such isomorphism is obtained by composing $\gamma_K$ with an element in $O(q_{\Lambda_{\hat g}})$. 

As shown above, the $O^+(\Gamma^{4,20})$ conjugacy classes with Frame shape $\pi_g$ are in one-to-one correspondence with classes of primitive embeddings $i:\Lambda_{\hat g}\hookrightarrow \Gamma^{4,20}$, modulo $O^+(\Gamma^{4,20})\times C_g(O(\Lambda_{\hat g}))$. We will prove that $O^+(\Gamma^{4,20})$  classes of embeddings are in one-to-one correspondence with pairs $(K,[t])$, where $K$ runs over the set of representatives of the  classes in $cl^+(\hat g)$, and $[t]$ is a double coset in $\overline{C}_{O(\Lambda_{\hat g})}(g)\backslash O(q_{\Lambda_{\hat g}})/\overline{O^+}(K)$. Given such a pair $(K,[t])$, choose an element $t\in O(q_{\Lambda_{\hat g}})$ in the double coset $[t]$ and consider the lattice
\be \Gamma_{K,t}:= \{(v,w)\in K^*\oplus \Lambda_{\hat g}^*\mid t\circ \gamma_K(\bar v)=\bar w\}\ . \ee Since $t\circ \gamma_K$ is an isomorphism $q_K\to -q_{\Lambda_{\hat g}}$, $\Gamma_{K,t}$ is an even unimodular lattice of signature $(4,20)$ and there is a sign structure preserving isomorphism  $\Gamma_{K,t}\cong \Gamma^{4,20}$. This isomorphism determines  a primitive embedding $i:\Lambda_{\hat g}\to \Gamma^{4,20}$ by $i(w)=(0,w)\in \Gamma_{K,t}\cong \Gamma^{4,20}$, which depends on the choice of the isomorphism $\Gamma_{K,t}\cong \Gamma^{4,20}$ only up to composition with $O^+(\Gamma^{4,20})$. Let $t'$ be a different element in the double coset $[t]$, so that $t'=\bar s t (\gamma_K \bar\sigma \gamma_K^{-1})$ for some $s\in C_g(O(\Lambda_{\hat g}))$ and $\sigma \in O^+(K)$. Then, there is an isomorphism $\Gamma_{K,t'}\to \Gamma_{K,t}\cong\Gamma^{4,20}$ given by $(v,w)\mapsto (\sigma(v),s^{-1}(w))$, so that the corresponding embeddings $i$ and $i'$ are related by $i(w)=i'(s(w))$. We conclude that there is a well-defined function mapping pairs $(K,[t])$ to classes of primitive embeddings modulo $O^+(\Gamma^{4,20})\times C_g(O(\Lambda_{\hat g}))$.

 Let us now consider a primitive embedding $i:\Lambda_{\hat g}\hookrightarrow \Gamma^{4,20}$. This embedding determines a description of $\Gamma^{4,20}$ as a sublattice of 
 $(\Gamma^g)^*\oplus i(\Lambda_{\hat g}^*)$:
 \be \Gamma^{4,20}=\{(v,i(w))\in (\Gamma^g)^*\oplus i(\Lambda_{\hat g}^*)\mid \gamma_i(\bar v)=\bar w\}\ ,
 \ee for some isomorphism $\gamma_i:q_{\Gamma^g} \stackrel{\cong}{\rightarrow} -q_{\Lambda_{\hat g}}$ of discriminant forms (see \eqref{gluing}). Since $\Gamma^g$ satisfies \eqref{genusGammag}, there exists a sign structure preserving isomorphism $\Gamma^g\cong K$ with one of the representatives of the classes in $cl^+(\hat g)$. This isomorphism determines an element $t\in O(q_{\Lambda_{\hat g}})$ such that $\gamma_i=t\gamma_K$; elements corresponding to different choices of the isomorphism are related by composition by $\overline{O}^+(K)$. If two primitive embeddings $i_1,i_2$ are related by $hi_1=i_2s$, for some $h\in O^+(\Gamma^{4,20})$ and $s\in C_g(O(\Lambda_{\hat g}))$, then there is a sign structure-preserving isomorphism $\sigma:\Gamma^{g_1}\stackrel{\cong}{\rightarrow}\Gamma^{g_2}$ given by
 \be h(v,0)=(\sigma(v),0),\qquad v\in {\Gamma^{g_1}}
 \ee and such that
 \be \bar s \gamma_{i_2} \bar \sigma =\gamma_{i_1}\ ,
 \ee where $\bar s$ and $\bar \sigma$ are the induced maps on the discriminant forms. Therefore, the elements $t_1,t_2\in O(q_{\Lambda_{\hat g}})$ such that $\gamma_{i_k}=t_k\gamma_K$ are related by 
 \be t_1=\bar st_2 (\gamma_K\bar\sigma\gamma_K^{-1})\ ,
 \ee so that $t_1$ and $t_2$ belong to the same double coset in $\overline{C}_{O(\Lambda_{\hat g})}(g)\backslash O(q_{\Lambda_{\hat g}})/\overline{O^+}(K)$. Therefore, there is a well-defined function that maps  classes of primitive embeddings into pairs $(K,[t])$, and this function is the inverse of the map defined above.
 \end{proof}

When the rank of $\Lambda_{\hat g}$ is exactly $20$, i.e. when the lattices in $cl^+(\hat g)$ are positive definite, the groups $O^+(K)$ and $C_g(O(\Lambda_{\hat g}))$ are finite and the number of classes can be computed directly.  This paper is accompanied by a text file containing the Magma program we wrote to perform this calculation. In the course of this calculation, we make use of results in \cite{HohnMason} classifying sublattices of the Leech lattice fixed by subgroups of $Co_0$.

When the rank of $\Lambda_{\hat g}$ is less than $20$, a brute force computation is not available, since the groups $O^+(K)$ have infinite order. Nevertheless, we can determine lower and upper bounds on the number of $O^+(\Gamma^{4,20})$ classes. Firstly, we know that, for each Frame shape, there is at least one $O^+(\Gamma^{4,20})$ class. Furthermore, when the corresponding twining genus has complex multiplier, the number of $O^+(\Gamma^{4,20})$ classes must be at least two.  Miranda and Morrison \cite{MirandaMorrison1,MirandaMorrison2} provide a practical algorithm to compute the number of right cosets $\sum_{K\in cl^+(g)} \left|  O(q_{\Lambda_{\hat g}})/\overline{O^+}(K)\right|$ in the case where $K$ is indefinite with rank at least three. This provides an upper bound on the number of $O^+(\Gamma^{4,20})$-classes in the case where $\Lambda_{\hat g}$ has rank less than $d$. This upper bound is almost always sharp -- it is either one or two depending on whether the twining genus has complex multiplier or not. The only exception is the Frame shape $1^{-4}2^53^46^1$, for which we were not able to determine whether the number of associated $O^+(\Gamma^{4,20})$ is one or two. On the other hand, even if there are two classes, they are necessarily the inverse of each other. This implies that the twining genera are the same. All results are collected in appendix \ref{a:results}.

\subsection{Twining Genera}\label{a:results}

In this subsection, we present the classification of four-plane-preserving $O^+(\Gamma^{4,20})$ conjugacy classes and our results on the corresponding twining genera. The results are summarized in table \ref{last_table_genera}. 

The first column $\pi_g$ contains the $42$ possible Frame shapes of four-plane-preserving classes of $O^+(\Gamma^{4,20})$, which are in one-to-one correspondence with the $42$ $Co_0$ conjugacy classes of automorphisms of the Leech lattice that fix a sublattice of rank at least $4$. 

For each  Frame shape $\pi_g$, the corresponding twining genera are weak Jacobi forms of weight $0$ and index $1$ for a subgroup $\modgrp\subseteq SL_2(\ZZ)$, defined in \eqref{modgrp}, with a multiplier $\psi$ of order $n$. In the second column we list the group $\modgrp$  and the order $n$ of the multiplier $\psi$.    We use the following notation to describe $\modgrp$.  If $\kappa$ is a subgroup of $(\mathbb{Z}/N\mathbb{Z})^\times$, then we define
$$\Gamma_\kappa(N)=\left\{\left.\left(\begin{array}{cc}a&b\\c&d\end{array}\right)\in \SL_2(\ZZ) \right| c\equiv0\text{ (mod $N$)},\quad a,d\text{ (mod $N$)}\in \kappa\right\}.$$
In this notation, the standard congruence subgroups $\Gamma_0(N)$ and $\Gamma_1(N)$ correspond to $\Gamma_\kappa(N)$ with  $\kappa=(\mathbb{Z}/N\mathbb{Z})^\times$ and $\kappa=\langle 1\rangle$ , respectively. 
Apart from these standard congruence subgroups, we also encounter groups with $\kappa =\langle -1\rangle = \{1,-1\}$.
We use the symbol $\Gamma_\kappa(N)_{|n}$ if the twining genus is a Jacobi form for the group $\Gamma_\kappa(N)$  with  multiplier of order $n>1$. When $n=1$, we simply write $\Gamma_\kappa(N)$. Note that, as discussed in appendix \ref{a:modularity}, in general specifying $\modgrp$ and $n$ is not sufficient to fix $\psi$ uniquely. 

The third and fourth columns report, respectively, the number of $O(\Gamma^{4,20})$ and $O^+(\Gamma^{4,20})$ conjugacy classes of each Frame shape. More precisely, in the third column, we put a symbol $\circ$ for each $O(\Gamma^{4,20})$ class. In the fourth column,  we put a symbol $\circ$ for each $O^+(\Gamma^{4,20})$-class that is fixed by world-sheet parity (i.e., it is a class also with respect to the full $O(\Gamma^{4,20})$ group) and a symbol $\updownarrow$ for each pair of $O^+(\Gamma^{4,20})$-classes that are exchanged under world-sheet parity (i.e., they merge to form a unique $O(\Gamma^{4,20})$ class). We are able to determine the number of such classes for all Frame shapes, except for $1^{-4}2^53^46^1$. In this case, there might be either a single class or two classes corresponding to inverse elements $g, g^{-1}\in O^+(\Gamma^{4,20})$. Notice that whenever $\modgrp$ is of the form $\gmo{N}$, there are always exactly two $O^+(\Gamma^{4,20})$ classes $[g],[g']$ (which may or may not be related by world-sheet parity), which are related by a power map, i.e. $[g']=[g^a]$ for some $a$ coprime to $N$. The corresponding twining genera are distinct, but are related by $\Gamma_0(N)$ transformations that are not in $\gmo{N}$. 

The fifth column reports, for each $O^+(\Gamma^{4,20})$ class, whether the corresponding twining genus $\eg_g$ is known in the following sense. A $\checkmark$ denotes a twining function which has been observed in an (IR) $K3$ NLSM. $LG$ denotes twining functions that have not been observed in an IR $K3$ NLSM but have been computed in \cite{cheng2015landau} as the twining genus of a UV symmetry in a LG orbifold which flows to a $K3$ NLSM in the IR, as discussed in \S \ref{LGtwin}. An $\times$ denotes a twining function which has not been observed in any $K3$ NLSM or LG orbifold anywhere in the literature. Finally, for those twining functions which have not been observed in a $K3$ NLSM, a $^\dagger$ denotes that nevertheless, the explicit twining function is fixed by the modularity arguments of \S \ref{s:twin_from_modul}.

In the last column, we report the list of Niemeier lattices $N$ such that the given Frame shape appears in the corresponding Niemeier group $G_N$. Equivalently, this is the list of those $N$ for which the corresponding Jacobi form $\phi^N_g$, arising from umbral or Conway moonshine is conjecturally equal to one of the twining genera $\eg_g$ of the given Frame shape. When $N=\Lambda$ and the $g$-invariant subspace has dimension exactly $4$,  we write two different symbols $\Lambda_+$ and $\Lambda_-$ to represent the two distinct genera $\phi_{g,+}^{\Lambda}$ $\phi_{g,-}^{\Lambda}$.   Niemeier lattices $N$ and $N'$ for which $\phi^N_g=\phi^{N'}_g$ are listed in the same row. More precisely, next to each  $O^+(\Gamma^{4,20})$-class for which the twining genus $\eg_g$ is known, we list all those $N$ for which $\phi^N_g=\eg_g$. In some cases, the same lattice $N$ appears in different rows for the same Frame shape: this occurs whenever two distinct $O^+(\Gamma^{4,20})$-classes have the same genus $\eg_g$.  Next to the $O^+(\Gamma^{4,20})$-classes for which the twining genus is unknown, we list  those $N$ for which, based on our conjectures,  $\phi^N_g$ is expected to coincide with $\eg_g$. For some Frame shapes, ($2^210^2$, $1^211^2$, $2^14^16^112^1$, $1^12^17^114^1$, $1^13^15^115^1$) we are not able to formulate any reasonable  conjecture associating $O^+(\Gamma^{4,20})$-classes with candidate twining genera $\phi_g^N$. In these cases, the alignment between classes and lists of Niemeier lattices has no meaning. This is represented by a square parenthesis in the third column.
A special case is the Frame shape $1^{-4}2^53^46^1$, to which there may correspond either one or two $O(\Gamma^{4,20})$ and $O^+(\Gamma^{4,20})$ classes; we emphasize our lack of certainty by writing $\circ,\circ^*$. If there are two, they are not related by world-sheet parity. However, the two classes are inverses of each other, so they have the same twining genus.

\newcolumntype{C}{>{$}c<{$}}
\rowcolors{2}{gray!11}{}
\begin{landscape}
\begin{tabularx}{\linewidth}{CCCCCC}
\pi_g & (\modgrp)_{|n}&    \begin{matrix}
O(\Gamma^{4,20})\\ \text{ classes}
\end{matrix} 
&    \begin{matrix}
O^+(\Gamma^{4,20})\\ \text{ classes}
\end{matrix} & \eg_g & \text{$\phi^N_g$} \\
\midrule\endhead
1^{24} & \SL_2(\ZZ)
&\begin{matrix}
\circ
\end{matrix} &\begin{matrix}
\circ
\end{matrix} 
 & \checkmark
 & \text{All}\\
 1^82^8 & \Gamma_0(2)  
&\begin{matrix}
\circ 
\end{matrix} &\begin{matrix}
\circ
\end{matrix} & \checkmark
& \text{All except $A_6^4,A_{12}^2,D_{12}^2,A_{24},D_{16}E_8,D_{24}$} 
\\

 1^{-8}2^{16} & \Gamma_0(2) &\begin{matrix}
\circ
\end{matrix} &\begin{matrix}
\circ
\end{matrix} & \checkmark 
&     \Lambda 
\\
2^{12} &\Gamma_0(2)_{|2}
&\begin{matrix}
\circ
\end{matrix}&\begin{matrix}
\circ
\end{matrix} & \checkmark& A_1^{24},A_2^{12},A_3^8,A_4^6,D_4^6,A_6^4,A_8^3,D_6^4,A_{12}^2,D_{12}^2,A_{24},\Lambda  
\\

 1^63^6 & \Gamma_0(3) & 
\begin{matrix}
\circ
\end{matrix} &\begin{matrix}
\circ
\end{matrix} & \checkmark
&
 A_1^{24},A_2^{12},A_3^8,A_5^4D_4,D_4^6,A_6^4,D_6^4,E_6^4,\Lambda
\\

 1^{-3}3^{9}  & \Gamma_0(3) &
 \begin{matrix}
\circ
\end{matrix} &\begin{matrix}
\circ
\end{matrix}  & \checkmark
&  \Lambda 
\\
 3^{8}    & \Gamma_0(3)_{|3} &\begin{matrix}
\circ
\end{matrix}   &
  \begin{matrix} \updownarrow\end{matrix}

 & \begin{matrix}
LG^\dagger
\\ \times^\dagger
\end{matrix} & \begin{matrix}
A_1^{24}, A_4^6, D_8^3 \\
A_2^{12}, D_4^6, A_8^3, E_8^3, \Lambda 
\end{matrix}
\\

 1^42^24^4 & \Gamma_0(4)
&\begin{matrix}
\circ
\end{matrix} &\begin{matrix}
\circ
\end{matrix} & \checkmark
&A_1^{24},A_2^{12},A_3^8,A_4^6,A_5^4D_4,D_4^6,A_9^2D_6,\Lambda 
\\

 1^82^{-8}4^8 & \Gamma_0(4)&\begin{matrix}
\circ
\end{matrix}   &\begin{matrix}
\circ
\end{matrix}  &  \times
& \Lambda
\\
 1^{-4}2^64^4 & \Gamma_0(4)  &\begin{matrix}
\circ
\end{matrix}  &\begin{matrix}
\circ
\end{matrix}  & \checkmark
& \Lambda
\\
 2^{-4}4^8  & \Gamma_0(4) 
&\begin{matrix}
\circ
\end{matrix}  & \begin{matrix} \updownarrow\end{matrix}  & \begin{matrix}\checkmark
 \\\checkmark
 \end{matrix}
 & \begin{matrix}
 \Lambda_+ \\ \Lambda_-
\end{matrix} 
\\
 2^{4}4^4  & \Gamma_0(4)_{|2} &\begin{matrix}
\circ
\end{matrix} &\begin{matrix}
\circ
\end{matrix}  & \checkmark
& A_1^{24},A_2^{12},A_3^8,D_4^6,A_7^2D_5^2,E_6^4,\Lambda  
\\

 4^6 & \Gamma_0(4)_{|4}
 &\begin{matrix}
\circ
\end{matrix} &   \begin{matrix} \updownarrow\end{matrix}  & 
\begin{matrix}
LG^\dagger\\ \times^\dagger
\end{matrix} 
 &\begin{matrix} A_1^{24}, A_2^{12}, A_6^4, D_6^4\\ A_3^8, A_4^6, A_{12}^2, \Lambda\end{matrix}
\\

 1^45^4 & \Gamma_0(5) 
&\begin{matrix}
\circ
\end{matrix} &\begin{matrix}
\circ
\end{matrix} & \checkmark
&A_1^{24},A_2^{12},A_4^6,D_4^6,\Lambda
\\

 1^{-1}5^5 & \gmo{5} &\begin{matrix}
\circ
\end{matrix} & \begin{matrix}\updownarrow\end{matrix}  & \begin{matrix}\checkmark
\\\checkmark
\end{matrix}   
 &  \begin{matrix}
\Lambda_+\\ \Lambda_-
\end{matrix}
\\
 1^22^23^26^2 & \Gamma_0(6)
&\begin{matrix}
\circ
\end{matrix} &\begin{matrix}
\circ
\end{matrix}  & \checkmark
&A_1^{24},A_2^{12},A_3^8,D_4^6,E_6^4,\Lambda
\\

 1^{4}2^13^{-4}6^5 & \Gamma_0(6) &\begin{matrix}
\circ
\end{matrix}  & \begin{matrix}
\circ
\end{matrix}   & \checkmark
 & \Lambda
\\
1^{5}2^{-4}3^16^4 & \Gamma_0(6) &\begin{matrix}
\circ
\end{matrix} & \begin{matrix}
\circ
\end{matrix}   & \checkmark
&\Lambda
\\
 1^{-2}2^43^{-2}6^4 & \Gamma_0(6)
&\begin{matrix}
\circ
\end{matrix}   & \begin{matrix} \updownarrow\end{matrix} 
 & \begin{matrix}\checkmark
 \\\checkmark
 \end{matrix} &  \begin{matrix}
\Lambda_+\\ \Lambda_-
\end{matrix}
\\
 1^{-1}2^{-1}3^36^3 & \Gamma_0(6)
&\begin{matrix}
\circ
\end{matrix}  & \begin{matrix} {\multirow{2}{*}{$\updownarrow$}}\\ {}\end{matrix} 
 & \begin{matrix}\checkmark
 \\\checkmark
 \end{matrix}
& \begin{matrix}
\Lambda_+\\ \Lambda_-
\end{matrix} \\
 1^{-4}2^53^46^1 & \Gamma_0(6) &\begin{matrix}
\circ,\circ^{*}
\end{matrix} & \begin{matrix}
\circ,\circ^{*}
\end{matrix} 
 & \checkmark
 & \Lambda\\
 2^36^3 & \Gamma_0(6)_{|2} &\begin{matrix}
\circ
\end{matrix}  & \begin{matrix}
\circ
\end{matrix}   & \checkmark
&  A_2^{12},A_3^8,A_6^4,\Lambda
\\ 
 6^{4} & \Gamma_0(6)_{|6}&\begin{matrix}
\circ\\\\\circ
\end{matrix}  & \begin{matrix} \updownarrow\\\\\updownarrow\end{matrix}   & \begin{matrix}
\times^\dagger\\LG
\\ \times^\dagger\\\times
\end{matrix}
& \begin{matrix}
A_1^{24}, A_4^6 \\ A_2^{12}, D_4^6,\Lambda_+ \\ A_1^{24}, A_4^6 \\ A_8^3, \Lambda_-
\end{matrix} 
\\

 1^37^3 & \Gamma_0(7) 
&\begin{matrix}
\circ
\end{matrix} &\begin{matrix}
\circ
\end{matrix}   & \checkmark
&A_1^{24},A_3^8,\Lambda
\\

 1^22^14^18^2 & \Gamma_0(8) 
&\begin{matrix}
\circ
\end{matrix} &\begin{matrix}
\circ
\end{matrix}   & \checkmark 
&A_1^{24},A_2^{12},A_5^4D_4,\Lambda
\\

 1^42^{-2}4^{-2}8^4 & \Gamma_0(8) &\begin{matrix}
\circ
\end{matrix}  &  \begin{matrix} {\multirow{2}{*}{$\updownarrow$}}\\ {}\end{matrix}  
  & \begin{matrix}\times\\\times\end{matrix}  
 &\begin{matrix}
\Lambda_+\\ \Lambda_-
\end{matrix} 
\\
 1^{-2}2^34^18^2 & \gmo{8} &\begin{matrix}
\circ
\end{matrix} &  \begin{matrix} \updownarrow\end{matrix}  & \begin{matrix}\checkmark
\\\checkmark
\end{matrix}
 &\begin{matrix}
\Lambda_+\\ \Lambda_-
\end{matrix} 
\\
 2^44^{-4}8^4  & \Gamma_0(8)_{|2} 
 & \begin{matrix} \circ \\ \circ\end{matrix}   & \begin{matrix} \circ \\ \circ\end{matrix}  & \begin{matrix}\checkmark\\\times\end{matrix} 
  & \begin{matrix}
\Lambda_+\\ \Lambda_-
\end{matrix}
\\
 4^28^2  & \Gamma_0(8)_{|4}&\begin{matrix}
\circ\\\\\circ
\end{matrix}    & \begin{matrix} \updownarrow\\\\\updownarrow\end{matrix}  & \begin{matrix}
\times^\dagger\\LG
\\\times^\dagger\\ \times
\end{matrix} 
& \begin{matrix}
A_2^{12}\\ A_3^8, \Lambda_+ \\A_2^{12}\\E_6^4, \Lambda_-
\end{matrix}
\\
 1^33^{-2}9^3  & \Gamma_0(9) & \begin{matrix} \circ \\ \circ \end{matrix}  & \begin{matrix} \circ \\ \circ \end{matrix}  & \begin{matrix}\checkmark
 \\\checkmark
 \end{matrix}
 & \begin{matrix}
\Lambda_+\\ \Lambda_-
\end{matrix} 
\\
 1^{2}2^15^{-2}10^3   & \gmo{10} 
&\begin{matrix}
\circ
\end{matrix}  & \begin{matrix} \updownarrow\end{matrix} & \begin{matrix}\checkmark
\\\checkmark
\end{matrix}
 & \begin{matrix}
\Lambda_+\\ \Lambda_-
\end{matrix} 
\\
 1^{3}2^{-2}5^110^2  & \gmo{10} 
&\begin{matrix}
\circ
\end{matrix}  & \begin{matrix} {}\\\updownarrow\\ {}\end{matrix} &\begin{matrix}\checkmark
\\\checkmark
\end{matrix} 
 & \begin{matrix}
\Lambda_+\\ \Lambda_-
\end{matrix} 
\\
 1^{-2}2^35^210^1  & \gmo{10} 
 &\begin{matrix}
\circ
\end{matrix}  & \begin{matrix} \updownarrow\end{matrix}
 & \begin{matrix}\checkmark
 \\\checkmark
 \end{matrix}
& \begin{matrix}
\Lambda_+\\ \Lambda_-
\end{matrix}
\\
 2^210^2 & \Gamma_0(10)_{|2}  
  & 
 \left[\;\begin{matrix} \circ\\  \circ \\  \circ \\ \circ \end{matrix}\;\right] & 
 \left[\;\begin{matrix} \updownarrow\\  \circ \\  \circ \\ \circ \end{matrix}\;\right] 
  & \begin{matrix}
LG
\\ \times
\end{matrix}
 & \begin{matrix}
A_1^{24}, A_2^{12},\Lambda_+ \\ A_4^6,\Lambda_-
\end{matrix}
\\

 1^211^2 & \Gamma_0(11)   & 
\left[\;\begin{matrix} \circ\\  \circ \\  \circ \end{matrix}\;\right]
 &\left[\;\begin{matrix} \updownarrow\\  \circ \\  \circ \end{matrix}\;\right]
 & \begin{matrix}
\times\\ LG
\end{matrix}
 & \begin{matrix}
 A_1^{24}, \Lambda_+ \\ A_2^{12},\Lambda_-
\end{matrix} 
\\
 1^{2}2^{-2}3^24^{2}6^{-2}12^2  & \Gamma_0(12)  & \begin{matrix}
\circ
\end{matrix}  & \begin{matrix} {\multirow{2}{*}{$\updownarrow$}}\\ {}\end{matrix}  & \begin{matrix}\times\\\times\end{matrix}  
 &  \begin{matrix}
\Lambda_+\\ \Lambda_-
\end{matrix}
\\
 1^{1}2^23^14^{-2}12^2  & \gmo{12} 
 & \begin{matrix}
\circ\\ \circ
\end{matrix}   & \begin{matrix}  \circ\\ \circ\end{matrix}  & \begin{matrix}\checkmark
\\\checkmark
\end{matrix}
  & \begin{matrix}
\Lambda_+\\ \Lambda_-
\end{matrix}
\\
1^{2}3^{-2}4^16^212^1  & \gmo{12} 
& \begin{matrix}
\circ\\ \circ
\end{matrix}  & \begin{matrix}
\circ\\ \circ
\end{matrix}  & \begin{matrix}\checkmark
\\\checkmark
\end{matrix} 
& \begin{matrix}
\Lambda_+\\ \Lambda_-
\end{matrix} 
\\
 1^{-2}2^23^24^112^1  & \gmo{12} 
& \begin{matrix}
\circ
\end{matrix}  & \begin{matrix} \updownarrow\end{matrix}  & \begin{matrix}\checkmark
\\\checkmark
\end{matrix}
& \begin{matrix}
\Lambda_+\\ \Lambda_-
\end{matrix} 
\\

 2^14^16^112^1 & \Gamma_0(12)_{|2}   
  &\left[\;\begin{matrix} \circ\\  \circ \\  \circ \end{matrix}\;\right]  & 
  \left[\;\begin{matrix} \updownarrow\\  \circ \\  \circ \end{matrix}\;\right]
  & \begin{matrix}
\times\\ \times
\end{matrix}
& \begin{matrix}
A_1^{24},\Lambda_+ \\ D_4^6, \Lambda_-
\end{matrix} 
 \\
 1^12^17^114^1 & \Gamma_0(14) 
 &\left[\;\begin{matrix} \circ\\  \circ \\  \circ \end{matrix}\;\right] 
 & \left[\;\begin{matrix} \updownarrow\\  \circ \\  \circ \end{matrix}\;\right]  & \begin{matrix}
\times\\ LG
\end{matrix}

&  \begin{matrix} A_1^{24},\Lambda_+ \\ A_3^8,\Lambda_-\end{matrix}
\\

 1^13^15^115^1  & \Gamma_0(15) & 
\left[\;\begin{matrix} \circ\\ \circ \\  \circ \end{matrix}\;\right]
 &\left[\;\begin{matrix} \updownarrow\\  \circ \\  \circ \end{matrix}\;\right] 
 & \begin{matrix}
\times\\ LG
\end{matrix}
 
&  \begin{matrix}
 A_1^{24},\Lambda_+\\D_4^6, \Lambda_-
\end{matrix}
\\

\bottomrule \caption{\label{last_table_genera}\footnotesize{ 
The conjugacy classes, status of the twining genera, and the corresponding Niemeier moonshine, for all the four-plane preserving Frame shapes.}}
\end{tabularx}
\end{landscape}

\bibliographystyle{utphys}
\bibliography{Refs}

\providecommand{\href}[2]{#2}\begingroup\raggedright\begin{thebibliography}{10}

\bibitem{mukai1988finite}
S.~Mukai, ``{Finite groups of automorphisms of K3 surfaces and the Mathieu
  group},'' {\em Inventiones mathematicae} {\bfseries 94} no.~1, (1988)
  183--221.

\bibitem{Kondo}
S.~Kond{\=o} {\em et~al.}, ``{Niemeier lattices, Mathieu groups, and finite
  groups of symplectic automorphisms of $ K3 $ surfaces},'' {\em Duke
  Mathematical Journal} {\bfseries 92} no.~3, (1998) 593--603.

\bibitem{K3symm}
M.~R. Gaberdiel, S.~Hohenegger, and R.~Volpato, ``{Symmetries of K3 sigma
  models},'' \href{http://dx.doi.org/10.4310/CNTP.2012.v6.n1.a1}{{\em
  Commun.Num.Theor.Phys.} {\bfseries 6} (2012) 1--50},
\href{http://arxiv.org/abs/1106.4315}{{\ttfamily arXiv:1106.4315 [hep-th]}}.
%%CITATION = ARXIV:1106.4315;%%.

\bibitem{huybrechts}
D.~Huybrechts, ``{On derived categories of K3 surfaces, symplectic
  automorphisms and the Conway group},'' {\em arXiv preprint arXiv:1309.6528}
  (2013) .

\bibitem{EqK3}
M.~C.~N. Cheng, J.~F.~R. Duncan, S.~M. Harrison, and S.~Kachru, ``{Equivariant
  K3 Invariants},''
\href{http://arxiv.org/abs/1508.02047}{{\ttfamily arXiv:1508.02047 [hep-th]}}.
%%CITATION = ARXIV:1508.02047;%%.

\bibitem{Cheng:2013wca}
M.~C. Cheng, J.~F. Duncan, and J.~A. Harvey, ``{Umbral moonshine and the
  Niemeier lattices},'' \href{http://dx.doi.org/10.1186/2197-9847-1-3}{{\em
  Research in the Mathematical Sciences} {\bfseries 1} no.~1, (2014) 1--81},
  \href{http://arxiv.org/abs/1307.5793}{{\ttfamily arXiv:1307.5793 [math.RT]}}.
\url{http://dx.doi.org/10.1186/2197-9847-1-3}.
%%CITATION = ARXIV:1307.5793;%%.

\bibitem{Aspinwall:1995zi}
P.~S. Aspinwall, ``{Enhanced gauge symmetries and K3 surfaces},''
  \href{http://dx.doi.org/10.1016/0370-2693(95)00957-M}{{\em Phys. Lett.}
  {\bfseries B357} (1995) 329--334},
\href{http://arxiv.org/abs/hep-th/9507012}{{\ttfamily arXiv:hep-th/9507012
  [hep-th]}}.
%%CITATION = HEP-TH/9507012;%%.

\bibitem{Witten:1995ex}
E.~Witten, ``{String theory dynamics in various dimensions},''
  \href{http://dx.doi.org/10.1016/0550-3213(95)00158-O}{{\em Nucl. Phys.}
  {\bfseries B443} (1995) 85--126},
\href{http://arxiv.org/abs/hep-th/9503124}{{\ttfamily arXiv:hep-th/9503124
  [hep-th]}}.
%%CITATION = HEP-TH/9503124;%%.

\bibitem{EOT}
T.~Eguchi, H.~Ooguri, and Y.~Tachikawa, ``{Notes on the K3 Surface and the
  Mathieu group $M_{24}$},''
  \href{http://dx.doi.org/10.1080/10586458.2011.544585}{{\em Exper.Math.}
  {\bfseries 20} (2011) 91--96},
\href{http://arxiv.org/abs/1004.0956}{{\ttfamily arXiv:1004.0956 [hep-th]}}.
%%CITATION = ARXIV:1004.0956;%%.

\bibitem{DMZ}
A.~Dabholkar, S.~Murthy, and D.~Zagier, ``{Quantum Black Holes, Wall Crossing,
  and Mock Modular Forms},''
\href{http://arxiv.org/abs/1208.4074}{{\ttfamily arXiv:1208.4074 [hep-th]}}.
%%CITATION = ARXIV:1208.4074;%%.

\bibitem{Cheng:2011ay}
M.~C. Cheng and J.~F. Duncan, ``{On Rademacher Sums, the Largest Mathieu Group,
  and the Holographic Modularity of Moonshine},''
  \href{http://dx.doi.org/10.4310/CNTP.2012.v6.n3.a4}{{\em
  Commun.Num.Theor.Phys.} {\bfseries 6} (2012) 697--758},
\href{http://arxiv.org/abs/1110.3859}{{\ttfamily arXiv:1110.3859 [math.RT]}}.
%%CITATION = ARXIV:1110.3859;%%.

\bibitem{Cheng:2012tq}
M.~C. Cheng, J.~F. Duncan, and J.~A. Harvey, ``{Umbral Moonshine},''
  \href{http://dx.doi.org/10.4310/CNTP.2014.v8.n2.a1}{{\em
  Commun.Num.TheorPhys.} {\bfseries 08} (2014) 101--242},
\href{http://arxiv.org/abs/1204.2779}{{\ttfamily arXiv:1204.2779 [math.RT]}}.
%%CITATION = ARXIV:1204.2779;%%.

\bibitem{Gannon:2012ck}
T.~Gannon, ``{Much ado about Mathieu},''
\href{http://arxiv.org/abs/1211.5531}{{\ttfamily arXiv:1211.5531 [math.RT]}}.
%%CITATION = ARXIV:1211.5531;%%.

\bibitem{Duncan:2015rga}
J.~F.~R. Duncan, M.~J. Griffin, and K.~Ono, ``{Proof of the Umbral Moonshine
  Conjecture},''
\href{http://arxiv.org/abs/1503.01472}{{\ttfamily arXiv:1503.01472 [math.RT]}}.
%%CITATION = ARXIV:1503.01472;%%.

\bibitem{Duncan:2014tya}
J.~F.~R. Duncan and J.~A. Harvey, ``{The Umbral Moonshine Module for the Unique
  Unimodular Niemeier Root System},''
\href{http://arxiv.org/abs/1412.8191}{{\ttfamily arXiv:1412.8191 [math.RT]}}.
%%CITATION = ARXIV:1412.8191;%%.

\bibitem{Gaberdiel:2012gf}
M.~R. Gaberdiel, D.~Persson, H.~Ronellenfitsch, and R.~Volpato, ``{Generalized
  Mathieu Moonshine},''
  \href{http://dx.doi.org/10.4310/CNTP.2013.v7.n1.a5}{{\em Commun.Num.Theor
  Phys.} {\bfseries 07} (2013) 145--223},
\href{http://arxiv.org/abs/1211.7074}{{\ttfamily arXiv:1211.7074 [hep-th]}}.
%%CITATION = ARXIV:1211.7074;%%.

\bibitem{Cheng:2016nto}
M.~C.~N. Cheng, P.~de~Lange, and D.~P.~Z. Whalen, ``{Generalised Umbral
  Moonshine},''
\href{http://arxiv.org/abs/1608.07835}{{\ttfamily arXiv:1608.07835 [math.RT]}}.
%%CITATION = ARXIV:1608.07835;%%.

\bibitem{Ooguri:1995wj}
H.~Ooguri and C.~Vafa, ``{Two-dimensional black hole and singularities of CY
  manifolds},'' \href{http://dx.doi.org/10.1016/0550-3213(96)00008-9}{{\em
  Nucl. Phys.} {\bfseries B463} (1996) 55--72},
\href{http://arxiv.org/abs/hep-th/9511164}{{\ttfamily arXiv:hep-th/9511164
  [hep-th]}}.
%%CITATION = HEP-TH/9511164;%%.

\bibitem{nikulin2013kahlerian}
V.~V. Nikulin, ``{K{\"a}hlerian K3 surfaces and Niemeier lattices. I},'' {\em
  Izvestiya: Mathematics} {\bfseries 77} no.~5, (2013) 954.

\bibitem{Cheng:2014zpa}
M.~C.~N. Cheng and S.~Harrison, ``{Umbral Moonshine and K3 Surfaces},''
  \href{http://dx.doi.org/10.1007/s00220-015-2398-5}{{\em Commun. Math. Phys.}
  {\bfseries 339} no.~1, (2015) 221--261},
\href{http://arxiv.org/abs/1406.0619}{{\ttfamily arXiv:1406.0619 [hep-th]}}.
%%CITATION = ARXIV:1406.0619;%%.

\bibitem{duncan2016derived}
J.~F. Duncan and S.~Mack-Crane, ``{Derived equivalences of K3 surfaces and
  twined elliptic genera},'' {\em Research in the Mathematical Sciences}
  {\bfseries 3} no.~1, (2016) 1--47.

\bibitem{frenkel1985moonshine}
I.~B. Frenkel, J.~Lepowsky, and A.~Meurman, ``A moonshine module for the
  monster,'' in {\em Vertex operators in mathematics and physics},
  pp.~231--273.
\newblock Springer, 1985.

\bibitem{duncan2007super}
J.~F. Duncan, ``Super-moonshine for {Conway}'s largest sporadic group,'' {\em
  Duke Mathematical Journal} {\bfseries 139} no.~2, (2007) 255--315.

\bibitem{Duncan:2014eha}
J.~F.~R. Duncan and S.~Mack-Crane, ``{The Moonshine Module for Conway's
  Group},''
\href{http://arxiv.org/abs/1409.3829}{{\ttfamily arXiv:1409.3829 [math.RT]}}.
%%CITATION = ARXIV:1409.3829;%%.

\bibitem{cheng2015landau}
M.~C.~N. Cheng, F.~Ferrari, S.~M. Harrison, and N.~M. Paquette,
  ``{Landau-Ginzburg Orbifolds and Symmetries of K3 CFTs},''
\href{http://arxiv.org/abs/1512.04942}{{\ttfamily arXiv:1512.04942 [hep-th]}}.
%%CITATION = ARXIV:1512.04942;%%.

\bibitem{Aspinwall:1996mn}
P.~S. Aspinwall, ``{K3 surfaces and string duality},'' in {\em {Fields, strings
  and duality. Proceedings, Summer School, Theoretical Advanced Study Institute
  in Elementary Particle Physics, TASI'96, Boulder, USA, June 2-28, 1996}},
  pp.~421--540.
\newblock 1996.
\newblock
\href{http://arxiv.org/abs/hep-th/9611137}{{\ttfamily arXiv:hep-th/9611137
  [hep-th]}}.
\newblock
%%CITATION = HEP-TH/9611137;%%.

\bibitem{Nahm:1999ps}
W.~Nahm and K.~Wendland, ``{A Hiker's guide to K3: Aspects of N=(4,4)
  superconformal field theory with central charge c = 6},''
  \href{http://dx.doi.org/10.1007/PL00005548}{{\em Commun. Math. Phys.}
  {\bfseries 216} (2001) 85--138},
\href{http://arxiv.org/abs/hep-th/9912067}{{\ttfamily arXiv:hep-th/9912067
  [hep-th]}}.
%%CITATION = HEP-TH/9912067;%%.

\bibitem{MR2376815}
T.~Bridgeland, ``{Stability conditions on $K3$ surfaces},''
  \href{http://dx.doi.org/10.1215/S0012-7094-08-14122-5}{{\em Duke Math. J.}
  {\bfseries 141} no.~2, (02, 2008) 241--291}.
  \url{http://dx.doi.org/10.1215/S0012-7094-08-14122-5}.

\bibitem{eichler_zagier}
M.~Eichler and D.~Zagier, {\em {The theory of Jacobi forms}}.
\newblock Birkh{\"a}user, 1985.

\bibitem{Troost:2010ud}
J.~Troost, ``{The non-compact elliptic genus: mock or modular},''
  \href{http://dx.doi.org/10.1007/JHEP06(2010)104}{{\em JHEP} {\bfseries 06}
  (2010) 104},
\href{http://arxiv.org/abs/1004.3649}{{\ttfamily arXiv:1004.3649 [hep-th]}}.
%%CITATION = ARXIV:1004.3649;%%.

\bibitem{Eguchi:2010cb}
T.~Eguchi and Y.~Sugawara, ``{Non-holomorphic Modular Forms and SL(2,R)/U(1)
  Superconformal Field Theory},''
  \href{http://dx.doi.org/10.1007/JHEP03(2011)107}{{\em JHEP} {\bfseries 03}
  (2011) 107},
\href{http://arxiv.org/abs/1012.5721}{{\ttfamily arXiv:1012.5721 [hep-th]}}.
%%CITATION = ARXIV:1012.5721;%%.

\bibitem{Ashok:2011cy}
S.~K. Ashok and J.~Troost, ``{A Twisted Non-compact Elliptic Genus},''
  \href{http://dx.doi.org/10.1007/JHEP03(2011)067}{{\em JHEP} {\bfseries 03}
  (2011) 067},
\href{http://arxiv.org/abs/1101.1059}{{\ttfamily arXiv:1101.1059 [hep-th]}}.
%%CITATION = ARXIV:1101.1059;%%.

\bibitem{Ashok:2013pya}
S.~K. Ashok, N.~Doroud, and J.~Troost, ``{Localization and real Jacobi
  forms},'' \href{http://dx.doi.org/10.1007/JHEP04(2014)119}{{\em JHEP}
  {\bfseries 04} (2014) 119},
\href{http://arxiv.org/abs/1311.1110}{{\ttfamily arXiv:1311.1110 [hep-th]}}.
%%CITATION = ARXIV:1311.1110;%%.

\bibitem{Murthy:2013mya}
S.~Murthy, ``{A holomorphic anomaly in the elliptic genus},''
  \href{http://dx.doi.org/10.1007/JHEP06(2014)165}{{\em JHEP} {\bfseries 06}
  (2014) 165},
\href{http://arxiv.org/abs/1311.0918}{{\ttfamily arXiv:1311.0918 [hep-th]}}.
%%CITATION = ARXIV:1311.0918;%%.

\bibitem{Harvey:2014nha}
J.~A. Harvey, S.~Lee, and S.~Murthy, ``{Elliptic genera of ALE and ALF
  manifolds from gauged linear sigma models},''
  \href{http://dx.doi.org/10.1007/JHEP02(2015)110}{{\em JHEP} {\bfseries 02}
  (2015) 110},
\href{http://arxiv.org/abs/1406.6342}{{\ttfamily arXiv:1406.6342 [hep-th]}}.
%%CITATION = ARXIV:1406.6342;%%.

\bibitem{Gaberdiel:2012um}
M.~R. Gaberdiel and R.~Volpato, ``{Mathieu Moonshine and Orbifold K3s},''
  \href{http://dx.doi.org/10.1007/978-3-662-43831-2_5}{{\em
  Contrib.Math.Comput.Sci.} {\bfseries 8} (2014) 109--141},
\href{http://arxiv.org/abs/1206.5143}{{\ttfamily arXiv:1206.5143 [hep-th]}}.
%%CITATION = ARXIV:1206.5143;%%.

\bibitem{Cheng:2010pq}
M.~C. Cheng, ``{K3 Surfaces, N=4 Dyons, and the Mathieu Group M24},''
  \href{http://dx.doi.org/10.4310/CNTP.2010.v4.n4.a2}{{\em
  Commun.Num.Theor.Phys.} {\bfseries 4} (2010) 623--658},
\href{http://arxiv.org/abs/1005.5415}{{\ttfamily arXiv:1005.5415 [hep-th]}}.
%%CITATION = ARXIV:1005.5415;%%.

\bibitem{Sen:2010ts}
A.~Sen, ``{Discrete Information from CHL Black Holes},''
  \href{http://dx.doi.org/10.1007/JHEP11(2010)138}{{\em JHEP} {\bfseries 1011}
  (2010) 138},
\href{http://arxiv.org/abs/1002.3857}{{\ttfamily arXiv:1002.3857 [hep-th]}}.
%%CITATION = ARXIV:1002.3857;%%.

\bibitem{Creutzig:2013mqa}
T.~Creutzig and G.~H{\"o}hn, ``{Mathieu Moonshine and the Geometry of K3
  Surfaces},'' \href{http://dx.doi.org/10.4310/CNTP.2014.v8.n2.a3}{{\em Commun.
  Num. Theor. Phys.} {\bfseries 08} (2014) 295--328},
\href{http://arxiv.org/abs/1309.2671}{{\ttfamily arXiv:1309.2671 [math.QA]}}.
%%CITATION = ARXIV:1309.2671;%%.

\bibitem{Volpato:2014zla}
R.~Volpato, ``{On symmetries of $\mathcal{N}=(4,4)$ sigma models on $T^4$},''
  \href{http://dx.doi.org/10.1007/JHEP08(2014)094}{{\em JHEP} {\bfseries 1408}
  (2014) 094},
\href{http://arxiv.org/abs/1403.2410}{{\ttfamily arXiv:1403.2410 [hep-th]}}.
%%CITATION = ARXIV:1403.2410;%%.

\bibitem{Gepner:1987qi}
D.~Gepner, ``{Space-Time Supersymmetry in Compactified String Theory and
  Superconformal Models},''
\href{http://dx.doi.org/10.1016/0550-3213(88)90397-5}{{\em Nucl. Phys.}
  {\bfseries B296} (1988) 757}.
%%CITATION = NUPHA,B296,757;%%.

\bibitem{Font:1989qc}
A.~Font, L.~E. Ibanez, and F.~Quevedo, ``{String Compactifications and $N=2$
  Superconformal Coset Constructions},''
\href{http://dx.doi.org/10.1016/0370-2693(89)91054-X}{{\em Phys. Lett.}
  {\bfseries B224} (1989) 79--88}.
%%CITATION = PHLTA,B224,79;%%.

\bibitem{Witten:1993jg}
E.~Witten, ``{On the Landau-Ginzburg description of N=2 minimal models},''
  \href{http://dx.doi.org/10.1142/S0217751X9400193X}{{\em Int. J. Mod. Phys.}
  {\bfseries A9} (1994) 4783--4800},
\href{http://arxiv.org/abs/hep-th/9304026}{{\ttfamily arXiv:hep-th/9304026
  [hep-th]}}.
%%CITATION = HEP-TH/9304026;%%.

\bibitem{phases}
E.~Witten, ``{Phases of N=2 theories in two-dimensions},''
  \href{http://dx.doi.org/10.1016/0550-3213(93)90033-L}{{\em Nucl. Phys.}
  {\bfseries B403} (1993) 159--222},
\href{http://arxiv.org/abs/hep-th/9301042}{{\ttfamily arXiv:hep-th/9301042
  [hep-th]}}.
%%CITATION = HEP-TH/9301042;%%.

\bibitem{Witten:1995zh}
E.~Witten, ``{Some comments on string dynamics},'' in {\em {Future perspectives
  in string theory. Proceedings, Conference, Strings'95, Los Angeles, USA,
  March 13-18, 1995}}, pp.~501--523.
\newblock 1995.
\newblock
\href{http://arxiv.org/abs/hep-th/9507121}{{\ttfamily arXiv:hep-th/9507121
  [hep-th]}}.
\newblock
%%CITATION = HEP-TH/9507121;%%.

\bibitem{Harvey:2013mda}
J.~A. Harvey and S.~Murthy, ``{Moonshine in Fivebrane Spacetimes},''
  \href{http://dx.doi.org/10.1007/JHEP01(2014)146}{{\em JHEP} {\bfseries 1401}
  (2014) 146},
\href{http://arxiv.org/abs/1307.7717}{{\ttfamily arXiv:1307.7717 [hep-th]}}.
%%CITATION = ARXIV:1307.7717;%%.

\bibitem{Harvey:2014cva}
J.~A. Harvey, S.~Murthy, and C.~Nazaroglu, ``{ADE Double Scaled Little String
  Theories, Mock Modular Forms and Umbral Moonshine},''
\href{http://arxiv.org/abs/1410.6174}{{\ttfamily arXiv:1410.6174 [hep-th]}}.
%%CITATION = ARXIV:1410.6174;%%.

\bibitem{Cheng:2013kpa}
M.~C. Cheng, X.~Dong, J.~Duncan, J.~Harvey, S.~Kachru, {\em et~al.}, ``{Mathieu
  Moonshine and N=2 String Compactifications},''
  \href{http://dx.doi.org/10.1007/JHEP09(2013)030}{{\em JHEP} {\bfseries 1309}
  (2013) 030},
\href{http://arxiv.org/abs/1306.4981}{{\ttfamily arXiv:1306.4981 [hep-th]}}.
%%CITATION = ARXIV:1306.4981;%%.

\bibitem{Harrison:2013bya}
S.~Harrison, S.~Kachru, and N.~M. Paquette, ``{Twining Genera of (0,4)
  Supersymmetric Sigma Models on K3},''
  \href{http://dx.doi.org/10.1007/JHEP04(2014)048}{{\em JHEP} {\bfseries 04}
  (2014) 048},
\href{http://arxiv.org/abs/1309.0510}{{\ttfamily arXiv:1309.0510 [hep-th]}}.
%%CITATION = ARXIV:1309.0510;%%.

\bibitem{Dijkgraaf:1996xw}
R.~Dijkgraaf, G.~W. Moore, E.~P. Verlinde, and H.~L. Verlinde, ``{Elliptic
  genera of symmetric products and second quantized strings},''
  \href{http://dx.doi.org/10.1007/s002200050087}{{\em Commun. Math. Phys.}
  {\bfseries 185} (1997) 197--209},
\href{http://arxiv.org/abs/hep-th/9608096}{{\ttfamily arXiv:hep-th/9608096
  [hep-th]}}.
%%CITATION = HEP-TH/9608096;%%.

\bibitem{Hoehn:2014ika}
G.~Hoehn and G.~Mason, ``{Finite groups of symplectic automorphisms of
  hyperk\"ahler manifolds of type $K3^{[2]}$},''
\href{http://arxiv.org/abs/1409.6055}{{\ttfamily arXiv:1409.6055 [math.AG]}}.
%%CITATION = ARXIV:1409.6055;%%.

\bibitem{Chaudhuri:1995fk}
S.~Chaudhuri, G.~Hockney, and J.~D. Lykken, ``{Maximally supersymmetric string
  theories in $D < 10$},''
  \href{http://dx.doi.org/10.1103/PhysRevLett.75.2264}{{\em Phys.Rev.Lett.}
  {\bfseries 75} (1995) 2264--2267},
\href{http://arxiv.org/abs/hep-th/9505054}{{\ttfamily arXiv:hep-th/9505054
  [hep-th]}}.
%%CITATION = HEP-TH/9505054;%%.

\bibitem{Chaudhuri:1995dj}
S.~Chaudhuri and D.~A. Lowe, ``{Type IIA heterotic duals with maximal
  supersymmetry},'' \href{http://dx.doi.org/10.1016/0550-3213(95)00589-7}{{\em
  Nucl.Phys.} {\bfseries B459} (1996) 113--124},
\href{http://arxiv.org/abs/hep-th/9508144}{{\ttfamily arXiv:hep-th/9508144
  [hep-th]}}.
%%CITATION = HEP-TH/9508144;%%.

\bibitem{Sen:1995ff}
A.~Sen and C.~Vafa, ``{Dual pairs of type II string compactification},''
  \href{http://dx.doi.org/10.1016/0550-3213(95)00498-H}{{\em Nucl.Phys.}
  {\bfseries B455} (1995) 165--187},
\href{http://arxiv.org/abs/hep-th/9508064}{{\ttfamily arXiv:hep-th/9508064
  [hep-th]}}.
%%CITATION = HEP-TH/9508064;%%.

\bibitem{Chaudhuri:1995bf}
S.~Chaudhuri and J.~Polchinski, ``{Moduli space of CHL strings},''
  \href{http://dx.doi.org/10.1103/PhysRevD.52.7168}{{\em Phys.Rev.} {\bfseries
  D52} (1995) 7168--7173},
\href{http://arxiv.org/abs/hep-th/9506048}{{\ttfamily arXiv:hep-th/9506048
  [hep-th]}}.
%%CITATION = HEP-TH/9506048;%%.

\bibitem{Dijkgraaf:1996it}
R.~Dijkgraaf, E.~P. Verlinde, and H.~L. Verlinde, ``{Counting dyons in N=4
  string theory},'' \href{http://dx.doi.org/10.1016/S0550-3213(96)00640-2}{{\em
  Nucl. Phys.} {\bfseries B484} (1997) 543--561},
\href{http://arxiv.org/abs/hep-th/9607026}{{\ttfamily arXiv:hep-th/9607026
  [hep-th]}}.
%%CITATION = HEP-TH/9607026;%%.

\bibitem{Shih:2005uc}
D.~Shih, A.~Strominger, and X.~Yin, ``{Recounting Dyons in N=4 string
  theory},'' \href{http://dx.doi.org/10.1088/1126-6708/2006/10/087}{{\em JHEP}
  {\bfseries 10} (2006) 087},
\href{http://arxiv.org/abs/hep-th/0505094}{{\ttfamily arXiv:hep-th/0505094
  [hep-th]}}.
%%CITATION = HEP-TH/0505094;%%.

\bibitem{Jatkar:2005bh}
D.~P. Jatkar and A.~Sen, ``{Dyon spectrum in CHL models},''
  \href{http://dx.doi.org/10.1088/1126-6708/2006/04/018}{{\em JHEP} {\bfseries
  04} (2006) 018},
\href{http://arxiv.org/abs/hep-th/0510147}{{\ttfamily arXiv:hep-th/0510147
  [hep-th]}}.
%%CITATION = HEP-TH/0510147;%%.

\bibitem{David:2006ji}
J.~R. David, D.~P. Jatkar, and A.~Sen, ``{Product representation of Dyon
  partition function in CHL models},''
  \href{http://dx.doi.org/10.1088/1126-6708/2006/06/064}{{\em JHEP} {\bfseries
  0606} (2006) 064},
\href{http://arxiv.org/abs/hep-th/0602254}{{\ttfamily arXiv:hep-th/0602254
  [hep-th]}}.
%%CITATION = HEP-TH/0602254;%%.

\bibitem{David:2006ru}
J.~R. David, D.~P. Jatkar, and A.~Sen, ``{Dyon Spectrum in N=4 Supersymmetric
  Type II String Theories},''
  \href{http://dx.doi.org/10.1088/1126-6708/2006/11/073}{{\em JHEP} {\bfseries
  11} (2006) 073},
\href{http://arxiv.org/abs/hep-th/0607155}{{\ttfamily arXiv:hep-th/0607155
  [hep-th]}}.
%%CITATION = HEP-TH/0607155;%%.

\bibitem{David:2006ud}
J.~R. David, D.~P. Jatkar, and A.~Sen, ``{Dyon spectrum in generic N=4
  supersymmetric Z(N) orbifolds},''
  \href{http://dx.doi.org/10.1088/1126-6708/2007/01/016}{{\em JHEP} {\bfseries
  01} (2007) 016},
\href{http://arxiv.org/abs/hep-th/0609109}{{\ttfamily arXiv:hep-th/0609109
  [hep-th]}}.
%%CITATION = HEP-TH/0609109;%%.

\bibitem{David:2006yn}
J.~R. David and A.~Sen, ``{CHL Dyons and Statistical Entropy Function from
  D1-D5 System},'' \href{http://dx.doi.org/10.1088/1126-6708/2006/11/072}{{\em
  JHEP} {\bfseries 0611} (2006) 072},
\href{http://arxiv.org/abs/hep-th/0605210}{{\ttfamily arXiv:hep-th/0605210
  [hep-th]}}.
%%CITATION = HEP-TH/0605210;%%.

\bibitem{Cheng:2008kt}
M.~C. Cheng and A.~Dabholkar, ``{Borcherds-Kac-Moody Symmetry of N=4 Dyons},''
  \href{http://dx.doi.org/10.4310/CNTP.2009.v3.n1.a2}{{\em
  Commun.Num.Theor.Phys.} {\bfseries 3} (2009) 59--110},
\href{http://arxiv.org/abs/0809.4258}{{\ttfamily arXiv:0809.4258 [hep-th]}}.
%%CITATION = ARXIV:0809.4258;%%.

\bibitem{Persson:2015jka}
D.~Persson and R.~Volpato, ``{Fricke S-duality in CHL models},''
  \href{http://dx.doi.org/10.1007/JHEP12(2015)156}{{\em JHEP} {\bfseries 12}
  (2015) 156},
\href{http://arxiv.org/abs/1504.07260}{{\ttfamily arXiv:1504.07260 [hep-th]}}.
%%CITATION = ARXIV:1504.07260;%%.

\bibitem{Taormina:2011rr}
A.~Taormina and K.~Wendland, ``{The overarching finite symmetry group of Kummer
  surfaces in the Mathieu group $M_{24}$},''
  \href{http://dx.doi.org/10.1007/JHEP08(2013)125}{{\em JHEP} {\bfseries 1308}
  (2013) 125},
\href{http://arxiv.org/abs/1107.3834}{{\ttfamily arXiv:1107.3834 [hep-th]}}.
%%CITATION = ARXIV:1107.3834;%%.

\bibitem{Taormina:2013mda}
A.~Taormina and K.~Wendland, ``{A twist in the M24 moonshine story},''
\href{http://arxiv.org/abs/1303.3221}{{\ttfamily arXiv:1303.3221 [hep-th]}}.
%%CITATION = ARXIV:1303.3221;%%.

\bibitem{Gaberdiel:2016iyz}
M.~R. Gaberdiel, C.~A. Keller, and H.~Paul, ``{Mathieu Moonshine and Symmetry
  Surfing},''
\href{http://arxiv.org/abs/1609.09302}{{\ttfamily arXiv:1609.09302 [hep-th]}}.
%%CITATION = ARXIV:1609.09302;%%.

\bibitem{Kachru:2016ttg}
S.~Kachru, N.~M. Paquette, and R.~Volpato, ``{3D String Theory and Umbral
  Moonshine},''
\href{http://arxiv.org/abs/1603.07330}{{\ttfamily arXiv:1603.07330 [hep-th]}}.
%%CITATION = ARXIV:1603.07330;%%.

\bibitem{M5}
M.~C.~N. Cheng, X.~Dong, J.~F.~R. Duncan, S.~Harrison, S.~Kachru, and T.~Wrase,
  ``{Mock Modular Mathieu Moonshine Modules},''
\href{http://arxiv.org/abs/1406.5502}{{\ttfamily arXiv:1406.5502 [hep-th]}}.
%%CITATION = ARXIV:1406.5502;%%.

\bibitem{Benjamin:2014kna}
N.~Benjamin, S.~M. Harrison, S.~Kachru, N.~M. Paquette, and D.~Whalen, ``{On
  the elliptic genera of manifolds of Spin(7) holonomy},''
  \href{http://dx.doi.org/10.1007/s00023-015-0454-5}{{\em Annales Henri
  Poincare} {\bfseries 17} no.~10, (2016) 2663--2697},
\href{http://arxiv.org/abs/1412.2804}{{\ttfamily arXiv:1412.2804 [hep-th]}}.
%%CITATION = ARXIV:1412.2804;%%.

\bibitem{Cheng:2015fha}
M.~C.~N. Cheng, S.~M. Harrison, S.~Kachru, and D.~Whalen, ``{Exceptional
  Algebra and Sporadic Groups at c=12},''
\href{http://arxiv.org/abs/1503.07219}{{\ttfamily arXiv:1503.07219 [hep-th]}}.
%%CITATION = ARXIV:1503.07219;%%.

\bibitem{Nikulin}
V.~V. Nikulin, ``Integer symmetric bilinear forms and some of their geometric
  applications,'' {\em Izv. Akad. Nauk SSSR Ser. Mat.} {\bfseries 43} no.~1,
  (1979) 111--177, 238.

\bibitem{ConwaySloane}
J.~H. Conway and N.~J.~A. Sloane,
  \href{http://dx.doi.org/10.1007/978-1-4757-6568-7}{{\em Sphere packings,
  lattices and groups}}, vol.~290 of {\em Grundlehren der Mathematischen
  Wissenschaften [Fundamental Principles of Mathematical Sciences]}.
\newblock Springer-Verlag, New York, third~ed., 1999.
\newblock \url{http://dx.doi.org/10.1007/978-1-4757-6568-7}.

\bibitem{Dijkgraaf:1989pz}
R.~Dijkgraaf and E.~Witten, ``{Topological Gauge Theories and Group
  Cohomology},''
\href{http://dx.doi.org/10.1007/BF02096988}{{\em Commun. Math. Phys.}
  {\bfseries 129} (1990) 393}.
%%CITATION = CMPHA,129,393;%%.

\bibitem{Dijkgraaf:1989hb}
R.~Dijkgraaf, C.~Vafa, E.~P. Verlinde, and H.~L. Verlinde, ``{The Operator
  Algebra of Orbifold Models},''
\href{http://dx.doi.org/10.1007/BF01238812}{{\em Commun. Math. Phys.}
  {\bfseries 123} (1989) 485}.
%%CITATION = CMPHA,123,485;%%.

\bibitem{Roche:1990hs}
P.~Roche, V.~Pasquier, and R.~Dijkgraaf, ``{QuasiHopf algebras, group
  cohomology and orbifold models},'' {\em Nucl. Phys. Proc. Suppl.} {\bfseries
  18B} (1990) 60--72.
[,60(1990)].
%%CITATION = NUPHZ,18B,60;%%.

\bibitem{HohnMason}
G.~Hoehn and G.~Mason, ``{The 290 fixed-point sublattices of the Leech
  lattice},''
\href{http://arxiv.org/abs/1505.06420}{{\ttfamily arXiv:1505.06420 [math.GR]}}.
%%CITATION = ARXIV:1505.06420;%%.

\bibitem{MirandaMorrison1}
R.~Miranda and D.~R. Morrison, ``The number of embeddings of integral quadratic
  forms. {I},'' {\em Proc. Japan Acad. Ser. A Math. Sci.} {\bfseries 61}
  no.~10, (1985) 317--320.
  \url{http://projecteuclid.org/euclid.pja/1195514534}.

\bibitem{MirandaMorrison2}
R.~Miranda and D.~R. Morrison, ``The number of embeddings of integral quadratic
  forms. {II},'' {\em Proc. Japan Acad. Ser. A Math. Sci.} {\bfseries 62}
  no.~1, (1986) 29--32. \url{http://projecteuclid.org/euclid.pja/1195514495}.

\end{thebibliography}\endgroup

\end{document}